\documentclass[12pt]{amsart}
\usepackage{amssymb,mathtools}
\usepackage{mathrsfs}
\usepackage{xypic}
\usepackage{bm}
\usepackage{mathabx}

\usepackage[colorlinks=true]{hyperref}

\theoremstyle{plain} 
\newtheorem{thm}[equation] {Theorem}
\newtheorem{cor}[equation]{Corollary}
\newtheorem{lem}[equation]{Lemma}
\newtheorem{prop}[equation]{Proposition}

\theoremstyle{definition}
\newtheorem{defn}[equation]{Definition}
\newtheorem{ex}[equation]{Example}

\theoremstyle{remark}
\newtheorem{rem}[equation]{Remark}

\numberwithin{equation}{section}
\numberwithin{figure}{section}
\numberwithin{table}{section}

\newcommand{\F}{\ensuremath{\mathbb{F}}}
\newcommand{\C}{\ensuremath{\mathbb{C}}}
\newcommand{\Q}{\ensuremath{\mathbb{Q}}}
\newcommand{\Z}{\ensuremath{\mathbb{Z}}}
\newcommand{\pair}[2]{\langle{#1} \, \vert \, {#2}\rangle}
\newcommand{\size}[1]{\lvert #1 \rvert}
\newcommand{\charac}[1]{\widehat{#1}}
\newcommand{\ra}{\rightarrow}
\newcommand{\ah}{\charac{A}}
\newcommand{\hr}{\charac{R}}
\DeclareMathOperator{\GL}{GL}
\newcommand{\al}{\alpha}
\newcommand{\be}{\beta}
\newcommand{\ga}{\gamma}
\DeclareMathOperator{\aut}{Aut}
\DeclareMathOperator{\Hom}{Hom}
\newcommand{\ft}[1]{\widehat{#1}}
\DeclareMathOperator{\ev}{eval}
\DeclareMathOperator{\id}{id}
\DeclareMathOperator{\ce}{ce}
\newcommand{\wh}{\textsc{h}}
\DeclareMathOperator{\hwe}{hwe}
\DeclareMathOperator{\stab}{Stab}
\DeclareMathOperator{\iso}{Isom}
\DeclareMathOperator{\im}{im}
\newcommand{\la}{\lambda}
\DeclareMathOperator{\FinAb}{\mathbf{FinAb}}

\title[Dualities]{Dualities for finite abelian groups \\ and applications to coding theory 
}
\author[J. A. Wood]{Jay A. Wood}
\address{Western Michigan University}
\email{jay.wood@wmich.edu}
\dedicatory{In memory of departed friends: \\
Russell R. Kieckhafer, \oldstylenums{1954}--\oldstylenums{2025} \\
Randy L. Koehler, \oldstylenums{1954}--\oldstylenums{2021} }
\subjclass[2010]{Primary: 20K01, 94B05}
\keywords{Duality, dual code, MacWilliams identities}

\begin{document}

\begin{abstract}
The choice of an isomorphism, a \emph{duality}, between a finite abelian group $A$ and its character group allows one to define dual codes of additive codes over $A$.  Properties of dualities and dual codes are studied, continuing work of Delsarte from 1973 and more recent work of Dougherty and his collaborators.
\end{abstract}

\maketitle
\section{Introduction}
There has been an increased interest in additive codes, and, with it, an increased interest in bringing to bear on additive codes some of the tools that are available for linear codes, such as dual codes and the MacWilliams identities.  This paper attempts to provide a unified account of how to do this, drawing on the work of many authors, especially Delsarte \cite{MR384310} and Dougherty and his collaborators, as well as some work of mine.  Along the way, corrections are provided for a few misconceptions that have appeared in the literature.  The paper has been written so as to be reasonably self-contained.

A \emph{duality} is an isomorphism $\phi: A \ra \ah$ between a finite abelian group $A$ and its character group $\ah$.  The choice of a duality is equivalent to the existence of a nondegenerate complex-valued inner product $\Phi: A \times A \ra \C^\times$.  Delsarte's paper \cite{MR384310} considers inner products that are symmetric, i.e., $\Phi(a,b) = \Phi(b,a)$ for $a,b \in A$, and, using an inner product, defines dual codes of additive codes over $A$, as well as establishing the size condition for dual codes, double duality, and the MacWilliams identities for the Hamming weight.

The present paper allows for nonsymmetric inner products.  Naturally associated to a duality $\phi: A \ra \ah$ is another duality $\phi^*: A \ra \ah$, a character-theoretic analogue of the transpose of a linear tranformation.  The associated inner products satisfy $\Phi^*(a,b) = \Phi(b,a)$ for $a, b \in A$, so that $\phi^* = \phi$ if and only if $\Phi$ is symmetric.
The inner product $\Phi$ provides two notions of orthogonality, which are the same when $\Phi$ is symmetric.  If $H$ is a subgroup of $A$, then there are left and right orthogonals defined by
\begin{align*}
\mathfrak{L}(H) &= \{ a \in A: \Phi(a,h)=1 , \text{ for all $h \in H$}\}, \\
\mathfrak{R}(H) &= \{ a \in A: \Phi(h,a)=1 , \text{ for all $h \in H$}\} .
\end{align*}
Again, the size condition holds, i.e., $\size{H} \cdot \size{\mathfrak{L}(H)} = \size{H} \cdot \size{\mathfrak{R}(H)} = \size{A}$, as does double duality: $\mathfrak{L}(\mathfrak{R}(H)) = H = \mathfrak{R}(\mathfrak{L}(H))$.  The MacWilliams identities hold for the complete and Hamming enumerators.

The idea of choosing different dualities as a way to define different dual codes of additive codes appears to have started with \cite{MR3789447}.  The present paper considers the set of all dualities of $A$, which is in one-to-one correspondence with the automorphism group of $A$.  The problem of how the dual codes of a subgroup depend on the choice of duality is intimately related to how the automorphism group $\aut(A)$ of $A$ acts on the subgroups of $A$.  For example, a subgroup has the same dual codes for every duality if and only if the subgroup is a characteristic subgroup.

There is a natural notion of congruence of dualities that generalizes congruence of symmetric blinear forms.  
Two dualities of $A$ are congruent when there exists an automorphism $\tau$ of $A$ so that the associated inner products satisfy $\Phi_2(a,b) = \Phi_1(a \tau, b \tau)$ for $a, b \in A$.  Roughly speaking, congruent dualities are the same up to a change of basis.

Here is a short guide to the paper.  Section~\ref{sect:CharGrps} presents features of character groups that are needed in subsequent sections, especially the fact that forming character groups is an exact contravariant functor whose square is the identity.  In Section~\ref{sec:Dualities}, dualities and their associated inner products are defined.  Additive codes and their dual codes are discussed in Section~\ref{sec:AddDualCodes}, including the size condition and double duality.  The dependence of dual codes on the choice of duality is explored in Section~\ref{sec:StructureQ}, congruence is discussed in Section~\ref{sec:congruence}, and the MacWilliams identities for the complete and Hamming enumerators appear in Section~\ref{sec:MWIds}.  The paper concludes with a comment about dualities for finite rings in Section~\ref{sec:FinRings}, together with some topics for future work.

\subsubsection*{Acknowledgments}
I thank Steven T. Dougherty for discussing various aspects of his and his collaborators work with me, for encouraging me to write this paper, and for providing comments on an earlier version of the paper.  

This paper is dedicated to the memory of Russ Kieckhafer and Randy Koehler, my friends for more than 50 years.

\section{Character groups}  \label{sect:CharGrps}

Let $A$ be a finite abelian group.  
A \emph{character} $\pi$ of $A$ is a group homomorphism $\pi: A \ra \C^\times$, where $\C^\times$ is the multiplicative group of all nonzero complex numbers.  Writing the group operation of $A$ as addition, a character $\pi$ satisfies $\pi(a_1 + a_2) = \pi(a_1) \pi(a_2)$, for all $a_1, a_2 \in A$.

Write $\ah$ for the set of all characters of $A$, so that $\ah = \Hom_{\Z}(A, \C^\times)$; $\ah$ is itself a multiplicative  abelian group via $(\pi_1  \pi_2)(a) = \pi_1(a)  \pi_2(a)$ for $\pi_1, \pi_2 \in \ah$ and $a \in A$.  The identity element of $\ah$ is the trivial character, all of whose values equal $1$.  We call $\ah$ the \emph{character group} of $A$.  
We adopt the convention of writing the evaluation of a character $\pi \in \ah$ at an element $a \in A$ as $\pair{\pi}{a} = \pi(a) \in \C^\times$.  We will write $\pair{\pi}{a}_A$ if the group needs to be made clear.  Thus, the homomorphism property of a character and the definition of the group operation in $\ah$ have the form:
\begin{align} \label{eqn:CharHom}
\pair{\pi}{a_1+a_2} &= \pair{\pi}{a_1} \, \pair{\pi}{a_2}, \quad \pi \in \ah, \quad a_1, a_2 \in A ; \\
\pair{\pi_1 \pi_2}{a} &= \pair{\pi_1}{a} \, \pair{\pi_2}{a}, \quad \pi_1, \pi_2 \in \ah, \quad a \in A .  \label{eqn:CharProduct}
\end{align}

\begin{rem}
It is also possible to define characters as group homomorphisms $\varpi: A \ra \Q/\Z$, so that the character group is $\Hom_\Z(A, \Q/\Z)$, cf., \cite[\S 2.2]{wood:turkey}.  Because $A$ is finite, $\Hom_\Z(A, \Q/\Z)$ is isomorphic to $\Hom_\Z(A, \C^\times)$ via $\varpi \mapsto \pi$, with
\[  \pi(a) = \exp(2 \bm{\pi} i \varpi(a)), \quad a \in A , \]
where $\exp$ is the complex exponential function and $\bm{\pi}$ is the well-known constant.  One warning: in formulas such as \eqref{eqn:vanishingSums} below, it is vital that $\Hom_\Z(A, \C^\times)$ be used.
\end{rem}

The next several results summarize some of the fundamental properties of character groups, organized to get quickly to the heart of the matter.
The results are drawn from sources such as  \cite{serre:linear-rep, terras:fourier-analysis, wood:duality, wood:turkey}.

\begin{lem}  \label{lem:cyclicGroup}
If $A$ is a finite cyclic group, then $\ah \cong A$.  If $a \neq 0$, then there exists a character $\pi \in \ah$ with $\pi(a) \neq 1$.
\end{lem}

\begin{proof}
Let $m = \size{A}$, and let $\ga$ be a generator for $A$.  Fix a primitive $m$th root $\zeta_m$ of $1$ in $\C^\times$.  Any character $\pi \in \ah$ is completely determined by the value of $\pi(\ga)$, which is an $m$th root of $1$. Define a function $f: \Z/m\Z \ra \ah$, $j \mapsto f_j$, where $f_j(\ga) = \zeta_m^j$.  One verifies that $f$ is an isomorphism of groups.  As $A \cong \Z/m\Z$, we have $A \cong \ah$.

For any $k$ that is relatively prime to $m$, $\zeta_m^k$ is a primitive $m$th root of $1$.  Thus, the character $f_k: A \ra \C^\times$ is injective, so that $f_k(a) \neq 1$ for any $a \neq 0$.
\end{proof}

\begin{rem}  \label{rem:NonUniqueIsom}
The isomorphism of Lemma~\ref{lem:cyclicGroup} is not unique in general: it depends on the choices of a generator $\ga$ of $A$ and a primitive $m$th root $\zeta_m$.  This foreshadows Proposition~\ref{prop:DualitiesAut}.
\end{rem}

\begin{lem}  \label{lem:Products}
Let $A_1, A_2$ be finite abelian groups.  Then
\[  \charac{A_1 \times A_2} \cong \ah_1 \times \ah_2 . \]
\end{lem}

\begin{proof}
Given a character $\pi \in \charac{A_1 \times A_2}$, define $\pi_1 \in \ah_1$ and $\pi_2 \in \ah_2$ by 
\[  \pi_1(a_1) = \pi(a_1,0), \quad \pi_2(a_2) = \pi(0,a_2), \quad a_1 \in A_1, a_2 \in A_2 . \]
Conversely, given $\pi_1 \in \ah_1$ and $\pi_2 \in \ah_2$, define $\pi \in \charac{A_1 \times A_2}$ by
\[  \pi(a_1,a_2) = \pi_1(a_1) \pi_2(a_2), \quad (a_1,a_2) \in A_1 \times A_2 . \]
One verifies that these definitions yield homomorphisms that are inverses of each other.
\end{proof}

\begin{lem}  \label{lem:ExtendingCharacters}
Let $H \subseteq A$ be a subgroup of a finite abelian group $A$.  If $\theta \in \charac{H}$, then there exists a character $\pi \in \ah$ that extends $\theta$, i.e., $\pi(h) = \theta(h)$ for all $h \in H$.
\end{lem}

\begin{proof}
If $H = A$, there is nothing to prove.  If $H \neq A$, take any $g \in A$ with $g \not\in H$, and let $P$ be the subgroup of $A$ generated by $H$ and $g$; $\size{P} > \size{H}$ because $g \not\in H$.  We will extend $\theta$ to a character $\pi$ of $P$.

Let $m$ be the order of $g$, and denote by $\langle g \rangle$ the cyclic subgroup generated by $g$.  If $H \cap \langle g \rangle = \{0\}$, pick any $m$th root $\zeta$ of $1$ in $\C^\times$.  Defining
\begin{equation}  \label{eqn:extendingCharacters}
\pi(a) = \begin{cases}
\theta(a), & a \in H, \\
\zeta, & a=g ,
\end{cases}
\end{equation}
and extending as a homomorphism, we get a character $\pi$ of $P$ that extends $\theta$.

If $H \cap \langle g \rangle \neq \{0\}$, let $k$ be the smallest positive integer so that $H \cap \langle g \rangle = \langle kg \rangle$.  Pick any $k$th root $\zeta$ of $\theta(kg) \in \C^\times$.  Again use \eqref{eqn:extendingCharacters} and extend as a homomorphism to yield a well-defined character $\pi$ of $P$ that extends $\theta$.

If $P=A$, we are done.  Otherwise, repeat the process on $\pi \in \charac{P}$.  Because the subgroups increase strictly in size, only a finite number of repetitions are needed.
\end{proof}

\begin{prop}  \label{prop:SizeCharGroup}
Let $A$ be a finite abelian group.  Then $\ah \cong A$.  In particular, $\size{\ah} = \size{A}$.  If $a \neq 0$, then there exists a character $\pi \in \ah$ with $\pi(a) \neq 1$.
\end{prop}

\begin{proof}
The group $A$ is a product of cyclic groups of prime power order by the fundamental theorem of finite abelian groups  \cite[Chapter I, \S 10]{LANG}.  Then apply Lemmas~\ref{lem:cyclicGroup} and~\ref{lem:Products}.

If $a \neq 0$, let $H$ be the subgroup generated by $a$, and, by Lemma~\ref{lem:cyclicGroup}, let $\theta \in \charac{H}$ be a character such that $\theta(a) \neq 1$.  Then extend $\theta$ to $\pi \in \ah$, by Lemma~\ref{lem:ExtendingCharacters}.
\end{proof}

As in Remark~\ref{rem:NonUniqueIsom}, the isomorphism $A \cong \ah$ is generally not unique.

Given two finite abelian groups $A_1, A_2$ and a homomorphism $\al: A_1 \ra A_2$, there is an induced homomorphism $\al^*: \ah_2 \ra \ah_1$ defined by
\begin{equation}  \label{eqn:DualHom}
\pair{\al^*(\pi_2)}{a_1}_{A_1} = \pair{\pi_2}{\al(a_1)}_{A_2}, \quad a_1 \in A_1, \quad \pi_2 \in \ah_2 .
\end{equation}
If $\al$ is invertible, then one verifies that $(\al^*)^{-1} = (\al^{-1})^*$.

For any finite abelian group $A$, define a homomorphism $\ev$ from $A$ to its double character group $\Hom_{\Z}(\ah, \C^\times)$ by
\begin{equation}  \label{eqn:DefnEval}
\pair{\ev(a)}{\pi}_{\ah} = \pair{\pi}{a}_A , \quad a \in A, \quad \pi \in \ah .
\end{equation}
That is, $\ev(a)$ is the `evaluate at $a$' character of $\ah$.

\begin{prop}  \label{prop:IdentifyADoubleDual}
For any finite abelian group $A$, the homomorphism $\ev: A \ra \Hom_{\Z}(\ah, \C^\times)$ is an isomorphism.
\end{prop}

\begin{proof}
Consider the kernel of $\ev$.  If $a \in \ker\ev$, then $\pair{\pi}{a}=1$ for all $\pi \in \ah$.  By Proposition~\ref{prop:SizeCharGroup}, if $a \neq 0$, then there is a character $\pi$ with $\pi(a) \neq 1$.  Thus $\ker\ev = 0$, and $\ev$ is injective.  Proposition~\ref{prop:SizeCharGroup}, applied twice, implies $\size{A} = \size{\Hom_{\Z}(\ah, \C^\times)}$, so that $\ev$ is also surjective.
\end{proof}

Let $\FinAb$ be the category whose objects are all finite abelian groups and whose morphisms are group homomorphisms.  Define $\mathcal{F}: \FinAb \ra \FinAb$ by $\mathcal{F}(A) = \ah$ and $\mathcal{F}(\al) = \al^*$, where $\al: A_1 \ra A_2$ is a morphism.  The next result shows that $\mathcal{F}$ is a Morita duality functor; cf., \cite[Theorem~3.2]{wood:duality}.

\begin{prop}
As defined above, $\mathcal{F}: \FinAb \ra \FinAb$ is an exact contravariant functor such that $\mathcal{F}^2$ is naturally equivalent to the identity functor.
\end{prop}

\begin{proof}
By the definition of $\al^*$, the functor $\mathcal{F}$ is contravariant.  Of course, $\mathcal{F}^2(A) = \Hom_{\Z}(\ah, \C^\times)$.  One verifies, for finite abelian groups $A_1$, $A_2$ and morphism $\al: A_1 \ra A_2$, that the following diagram commutes:
\[  \xymatrix{A_1 \ar[r]^-{\ev} \ar[d]^-{\al} &  \Hom_{\Z}(\ah_1, \C^\times)  \ar[d]^-{\al^{**}}  \\
A_2 \ar[r]^-{\ev} & \Hom_{\Z}(\ah_2, \C^\times).  } \]
Indeed, for $a_1 \in A_1$, $\pi_2 \in \ah_2$, and using \eqref{eqn:DualHom} and \eqref{eqn:DefnEval}, we have
\begin{align*}
\pair{\ev(\al(a_1))}{\pi_2}_{\ah_2} &= \pair{\pi_2}{\al(a_1)}_{A_2} = \pair{\al^*(\pi_2)}{a_1}_{A_1} , \\
\pair{\al^{**}(\ev(a_1))}{\pi_2}_{\ah_2} &= \pair{\ev(a_1)}{\al^*(\pi_2)}_{\ah_1} 
= \pair{\al^*(\pi_2)}{a_1}_{A_1} .
\end{align*}

For exactness, take any short exact sequence of finite abelian groups
\[ \xymatrix{0 \ar[r] &H \ar[r]^{\al} &A \ar[r]^{\be} &Q \ar[r] &0} . \]
We need to show that the associated sequence
\begin{equation}  \label{eqn:DualExactSeq}
\xymatrix{1 &\charac{H} \ar[l] &\ah \ar[l]_{\al^*} &\charac{Q} \ar[l]_{\be^*}  &1 \ar[l]} 
\end{equation}
is also a short exact sequence.

Suppose $\pi \in \ker \be^*$.  This means $\pair{\be^*(\pi)}{a}_A = 1$ for all $a \in A$.  Then, $1 = \pair{\pi}{\be(a)}_Q$ for all $a \in A$.  Thus $\pi \in \charac{Q}$ is trivial, as $\be$ is surjective.

Because $\im \al \subseteq \ker \be$, we have $\im \be^* \subseteq \ker \al^*$.  Conversely, suppose $\pi \in \ker \al^*$.  This means, for any $h \in H$, $1 = \pair{\al^*(\pi)}{h}_H = \pair{\pi}{\al(h)}_A$.  Thus $\pi$ vanishes on $\im \al = \ker \be$.  This implies that $\pi$ descends to a well-defined character $\tilde{\pi}$ on $Q$: $\pair{\pi}{a}_A = \pair{\tilde{\pi}}{\be(a)}_Q = \pair{\be^*(\tilde{\pi})}{a}_A$.  Thus, $\pi = \be^*(\tilde\pi) \in \im \be^*$, and $\ker \al^* = \im \be^*$.

Finally, $\al^*$ is surjective by Lemma~\ref{lem:ExtendingCharacters}.
\end{proof}

From here on, we will identify $A$ and $\Hom_\Z(\ah, \C^\times)$ via $\ev$.
Using this identification, we have that
\begin{equation}  \label{eqn:DoubleDuals}
\al^{**} = \al ,
\end{equation}
for any homomorphism $\al: A_1 \ra A_2$.

If $H \subseteq A$ is a subgroup of a finite abelian group $A$, its \emph{annihilator} is the subgroup of $\ah$ defined by 
\begin{equation}  \label{eqn:DefnAnnihilator}
(\ah: H) = \{ \pi \in \ah: \pair{\pi}{h} = 1 \text{ for all $h \in H$} \} . 
\end{equation}

\begin{cor}  \label{cor:DoubleAnnihilator}
For a finite abelian group $A$ and a subgroup $H \subseteq A$, $\charac{A/H} \cong (\ah:H)$.  In particular, $\size{(\ah : H)} = \size{A}/\size{H}$.  Identifying $A$ and $\Hom_{\Z}(\ah, \C^\times)$ via $\ev$, we have 
\[  (A:(\ah:H)) = H . \]
\end{cor}

\begin{proof}
In the notation of \eqref{eqn:DualExactSeq}, $\charac{A/H} = \charac{Q} \cong \im \be^* = \ker \al^*$, but $\ker \al^* = (\ah:H)$.  The size statement now follows from Proposition~\ref{prop:SizeCharGroup}.

As for the double annihilator, $H \subseteq (A:(\ah:H))$ follows directly from the definition of $(\ah:H)$.  Equality then follows from the size statement, applied twice.
\end{proof}

\begin{prop}  \label{prop:VanishingSums}
For a finite abelian group $A$, subgroup $H \subseteq A$, and $\pi \in \ah$,
\begin{equation}  \label{eqn:vanishingSums}
\sum_{h \in H} \pair{\pi}{h} = \begin{cases}
\size{H}, & \pi \in (\ah : H) , \\
0, & \pi \not\in (\ah : H) .
\end{cases}
\end{equation}
Dually, for a subgroup $E \subseteq \ah$ and $a \in A$, 
\[  \sum_{\pi \in E} \pair{\pi}{a} = \begin{cases}
\size{E}, & a \in (A:E), \\
0, & a \not\in (A:E).
\end{cases}  \]
\end{prop}

\begin{proof}
If $\pi \in (\ah : H)$, then $\pair{\pi}{h} = 1$ for all $h \in H$; the sum equals $\size{H}$.  If $\pi \not\in (\ah : H)$, then there exists $h_0 \in H$ such that $\pair{\pi}{h_0} \neq 1$.  By reindexing the sum via $h = h_0+h'$, we see that
\begin{align*}
\sum_{h \in H} \pair{\pi}{h} &= \sum_{h' \in H} \pair{\pi}{h_0+h'} = \sum_{h' \in H} \pair{\pi}{h_0} \,  \pair{\pi}{h'} = \pi(h_0) \sum_{h' \in H} \pair{\pi}{h'} .
\end{align*}
As $\pair{\pi}{h_0} \neq 1$, the sum must vanish.
\end{proof}

By choosing $H=A$ and $E = \ah$ in Proposition~\ref{prop:VanishingSums}, and using Proposition~\ref{prop:SizeCharGroup}, we have the following corollary.

\begin{cor}  \label{cor:VanishSum}
Let $A$ be a finite abelian group.  For $\pi \in \ah$ and $a \in A$,
\[  \sum_{a \in A} \pair{\pi}{a} = \begin{cases}
\size{A}, & \pi=1, \\
0, & \pi \neq 1;
\end{cases}
\quad
\sum_{\pi \in \ah} \pair{\pi}{a} = \begin{cases}
\size{A}, & a=0, \\
0, & a \neq 0 .
\end{cases}  \]
\end{cor}

The fundamental theorem of finite abelian groups says that any finite abelian group can be written as a product of cyclic subgroups of prime power order.  The numbers and orders of the cyclic subgroups are uniquely determined, but the subgroups themselves are usually not.  For example, there are many choices of bases for a finite-dimensional vector space over a finite field $\F_p$ of dimension at least $2$.

There is a coarser decomposition of a finite abelian group, working prime by prime, that has the advantage of the component subgroups being unique.  

Let $A$ be a finite abelian group.  We know that the order $o(a)$ of any element $a \in A$ must divide $\size{A}$.  For every prime $p$ that divides $\size{A}$, define $A_p = \{ a \in A \colon o(a) = p^k \text{ for some integer $k$}\}$.  One shows that $A_p$ is a subgroup of $A$.    Let $\aut(A)$ be the group of automorphisms of $A$.

\begin{prop}  \label{prop:PrimeByPrime}
Let $A$ be a finite abelian group, and let $s(A)$ be the set of primes that divide $\size{A}$.  Then,
\begin{itemize}
\item  for $p \in s(A)$, $A_p$ is a $p$-group;
\item $A = \oplus_{p \in s(A)} A_p$;
\item  $\aut(A) = \oplus_{p \in s(A)} \aut(A_p)$.
\end{itemize}
\end{prop}

\begin{proof}
The first two items are Theorem~5 of \cite[Chapter I, \S 10]{LANG}.  
The decomposition of $\aut(A)$ follows from the observation that, for distinct primes $p \neq \ell$, any homomorphism $\al: A_p \ra A_\ell$ must be the zero-homomorphism, as $\size{\im \al}$ must divide both $\size{A_p}$ and $\size{A_\ell}$.
\end{proof}

Proposition~\ref{prop:PrimeByPrime} allows us to study a finite abelian group one prime at a time.

\section{Dualities and inner products}  \label{sec:Dualities}

In preparation for defining additive codes over a finite abelian group $A$ and their dual codes, we follow \cite{MR384310, MR3789447} and define dualities of $A$ and their associated inner products.

Let $A$ be a finite abelian group.  A \emph{duality} of $A$ is a group isomorphism $\phi: A \ra \ah$.  Let $\iso(A, \ah)$ be the set of all dualities of $A$.
Dualities exist by Proposition~\ref{prop:SizeCharGroup}, so $\iso(A, \ah)$ is nonempty.  As mentioned in Remark~\ref{rem:NonUniqueIsom}, there is generally more than one duality of $A$.  Proposition~\ref{prop:DualitiesAut} below makes this precise.

Suppose $\phi_0: A \ra \ah$ is a duality.  Define a map $f: \aut(A) \ra \iso(A, \ah)$, sending $\tau \in \aut(A)$ to the composition $\xymatrix{A \ar[r]^{\tau} & A \ar[r]^{\phi_0} & \ah}$.

\begin{prop}  \label{prop:DualitiesAut}
The map $f: \aut(A) \ra \iso(A, \ah)$ is a bijection. In particular, $\size{\iso(A,\ah)} = \size{\aut(A)}$.
\end{prop}

\begin{proof}
Define a map $g: \iso(A,\ah) \ra \aut(A)$ sending $\phi \in \iso(A, \ah)$ to the composition
$ \xymatrix{A \ar[r]^{\phi} & \ah \ar[r]^{\phi_0^{-1}} & A,} $
which is an automorphism $\tau \in \aut(A)$, with $\phi_0 \circ \tau = \phi$.  One verifies that $f$ and $g$ are inverses, hence bijections.
\end{proof}

Suppose $\phi: A \ra \ah$ is a duality of $A$.  Because $A$ and $\ah$ are both finite abelian groups and $\phi$ is a homomorphism between them, the induced homomorphism of \eqref{eqn:DualHom}, i.e.,
\[  \phi^*: \Hom_\Z(\ah, \C^\times)=A \ra \ah \]
is also a duality.   We say that a duality $\phi: A \ra \ah$ is \emph{symmetric} if $\phi^* = \phi$.  By \eqref{eqn:DoubleDuals}, we always have $\phi^{**} = \phi$ for any duality $\phi$.

\begin{lem}  \label{lem:SymDual}
Let $A$ be a finite abelian group, and let $\phi:A \ra \ah$ be a duality of $A$.  Then,
\[  \pair{\phi^*(a)}{b} = \pair{\phi(b)}{a} , \quad a, b \in A . \]
Thus, $\phi$ is symmetric if and only if $\pair{\phi(b)}{a} = \pair{\phi(a)}{b}$ for all $a, b \in A$.
\end{lem}

\begin{proof}
The key is to unravel the identification of $A$ and $\Hom_\Z(\ah, \C^\times)$ via $\ev$.  Using \eqref{eqn:DualHom} and \eqref{eqn:DefnEval}, we have, for all $a, b \in A$,
\begin{align*}
\pair{\phi^*(a)}{b}_A &= \pair{\phi^*(\ev(a))}{b}_A \\
&= \pair{\ev(a)}{\phi(b)}_{\ah} 
= \pair{\phi(b)}{a}_A  . \qedhere
\end{align*}
\end{proof}

\begin{lem}[{\cite[Corollary~4.2]{STD-DualitiesFinAbGp}}]  \label{lem:DualCyclic}
Let $A$ be a finite cyclic group.  Then every duality of $A$ is symmetric.
\end{lem}

\begin{proof}
Set $m = \size{A}$.  Let $\ga$ be a generator of $A$, and let $\zeta_m$ be a primitive $m$th root of $1$ in $\C$.  For every $j \in \Z/m\Z$, define a character $\pi_j \in \ah$ by $\pi_j(\ga^i) = \zeta_m^{ij}$, for $i \in \Z/m\Z$.  We saw in the proof of Lemma~\ref{lem:cyclicGroup} that every character of $A$ has this form.

Define $\phi_0: A \ra \ah$ by $\phi_0(\ga^i) = \pi_i$.  One verifies that $\phi_0$ is a duality.  Because $\pair{\pi_i}{\ga^j} = \zeta_m^{ij} = \pair{\pi_j}{\ga^i}$ for all $i,j \in \Z/m\Z$, Lemma~\ref{lem:SymDual} implies that $\phi_0$ is symmetric.  It is well-known that automorphisms of $A$ are induced by sending $\ga$ to $\ga^k$, where $k$ is relatively prime to $m$.  Thus $\phi(\ga^i) = \phi_0(\ga^{ki})$, $i \in \Z/m\Z$, defines another duality of $A$, and every duality of $A$ has this form, by Proposition~\ref{prop:DualitiesAut}.  Then $\pair{\phi(\ga^i)}{\ga^j} = \pair{\phi_0(\ga^{ki})}{\ga^j} = \zeta_m^{kij} = \pair{\phi_0(\ga^{kj})}{\ga^i} = \pair{\phi(\ga^j)}{\ga^i}$, and $\phi$ is symmetric by Lemma~\ref{lem:SymDual}.
\end{proof}

\begin{lem}  \label{lem:DualProduct}
Let $A_1, A_2$ be finite abelian groups.  If $\phi_1, \phi_2$ are symmetric dualities of $A_1, A_2$, respectively, then $\phi_1 \times \phi_2$ is a symmetric duality of $A_1 \times A_2$.
\end{lem}

\begin{proof}
Using Lemma~\ref{lem:Products}, one verifies the condition in Lemma~\ref{lem:SymDual}.
\end{proof}

\begin{prop}  \label{prop:SymDualsExist}
Let $A$ be a finite abelian group.  Then, there exists at least one symmetric duality of $A$.
\end{prop}

\begin{proof}
Write $A$ as a product of finite cyclic groups, by the fundamental theorem of finite abelian groups.  Then use Lemmas~\ref{lem:DualCyclic} and~\ref{lem:DualProduct}.
\end{proof}

\begin{rem}
Caveat!  The proof of Proposition~\ref{prop:SymDualsExist} does not imply that every duality of a finite abelian group is symmetric.  The reason is that Lemma~\ref{lem:DualProduct} applies only to dualities of a product that are in the `diagonal' form of $\phi_1 \times \phi_2$.  Especially important is the case where $A_1=A_2$, where there are more automorphisms of $A_1 \times A_2$ than just the diagonal ones.  This is discussed further in Example~\ref{ex:elem-ab-p-groups}.
\end{rem}

\begin{ex}  \label{ex:elem-ab-p-groups}
Let $p$ be a prime, and suppose $A$ is an elementary abelian $p$-group of order $p^n$.  Then $A$ is isomorphic to the underlying abelian group of a vector space of dimension $n$ over the finite field $\F_p$.   Elements of $A$ will be viewed as row vectors $a =[ a_1, a_2, \ldots, a_n ]$, with each $a_i \in \F_p \cong \Z/p\Z$.   Automorphisms of $A$ are given by invertible $n \times n$ matrices over $\F_p$ acting on $A$ on the right by matrix multiplication; i.e., $\aut(A) = \GL(n,\F_p)$.

Pick a primitive $p$th root $\zeta_p$ of $1$ in $\C$.  For $a \in A$, define $\pi_a \in \ah$ by
\[  \pair{\pi_a}{b} = \zeta_p^{a b^\top} \in \C^\times, \quad  b \in A . \]
Then $\phi_0: A \ra \ah$, $\phi_0(a) = \pi_a$, is a symmetric duality.  By Proposition~\ref{prop:DualitiesAut}, every other duality $\phi: A \ra \ah$ has the form $\phi(a) = \phi_0(aP)$, where $P \in \GL(n,\F_p)$.  Thus $\pair{\phi(a)}{b} = \pair{\pi_{aP}}{b} = \zeta_p^{a P b^\top}$.  Then $\phi^*: A \ra \ah$ is given by
\[  \pair{\phi^*(a)}{b} = \pair{\phi(b)}{a} = \zeta_p^{b P a^\top} = \zeta_p^{a P^\top b^\top} , \]
where have used the fact that $b P a^\top$ is a $1 \times 1$ matrix, so it equals its own transpose.  Thus $\phi^*(a) = \pi_{aP^\top}$, and $\phi$ is symmetric if and only if $P$ is symmetric.
This characterization of symmetric dualities is contrary to that in \cite[Theorem~2.5]{STD-DualitiesFinAbGp} and \cite[Lemma~4]{MR4770737}; corrections to the latter appear in  \cite{STD:SymDualities}.
\end{ex}

Is being symmetric common or rare?  The next result says, at least over vector spaces, that symmetric dualities are asymptotically rare.  This result also appears, independently, in \cite{STD:SymDualities}.
\begin{prop}
Let $A = \F_p^n$.  For a fixed prime $p$, the probability that a duality of $A$ is symmetric goes to $0$ as $n \ra \infty$.  Similarly, for a fixed $n$, the probability that a duality of $A$ is symmetric goes to $0$ as primes $p \ra \infty$.
\end{prop}

\begin{proof}
The probability that a duality $\phi: A \ra \ah$ is symmetric is
\[  \frac{\size{\{ P \in \GL(n,\F_p): P = P^\top \}}}{\size{\GL(n,\F_p)}} . \]
MacWilliams \cite[p.\ 156]{MR238870} gives the number $N(n)$ of symmetric, invertible $n \times n$ matrices over any finite field $\F_q$:
\begin{equation}  \label{eqn:SymInvertCount}
N(2t) = \prod_{i=1}^t (q^{2t+1} - q^{2i}), \quad
N(2t + 1) = \prod_{i=0}^t (q^{2t+1} - q^{2i}).
\end{equation}
The number of invertible matrices is
\[  \size{\GL(n,\F_q)} = \prod_{i=0}^{n-1} (q^n - q^i) . \]
For $n$ even or odd, i.e., for $n=2t$ or $n=2t+1$, respectively, we have
\begin{align*}
\frac{N(2t)}{\size{\GL(2t, \F_q)}} &= \frac{q^t}{\prod_{j=0}^{t-1} (q^{2t} - q^{2j})} < \frac{q^t}{q^{2t}-1} , \\
\frac{N(2t+1)}{ \size{\GL(2t+1,\F_q)}} &= \frac{1}{\prod_{j=1}^t (q^{2t+1} - q^{2j-1})} < \frac{1}{q^{2t+1}-1} .
\end{align*}
In both cases the ratio goes to $0$ for fixed $q \geq 2$ as $t \ra \infty$ (or for fixed $t \geq 1$ as $q \ra \infty$).  Of course, the same is true when we restrict $q$ to be a prime $p$.
\end{proof}

We conclude this section by describing inner products on a finite abelian group $A$, and we show that dualities on $A$ are equivalent to inner products on $A$.  Almost all of this material can be found in
Delsarte \cite[\S 6.1]{MR384310}.  

Let $A$ be a finite abelian group.  A function $\Psi: A \times A \ra \C^\times$ is an \emph{inner product} on $A$ if it satisfies the following properties:
\begin{itemize}
\item $\Psi(a_1+a_2,b) = \Psi(a_1,b) \Psi(a_2,b)$, for all $a_1, a_2, b \in A$;
\item $\Psi(a,b_1+b_2) = \Psi(a,b_1) \Psi(a,b_2)$, for all $a, b_1, b_2 \in A$;
\item  if $\Psi(a,b)=1$ for all $b \in A$, then $a=0$;
\item  if $\Psi(a,b)=1$ for all $a \in A$, then $b=0$.
\end{itemize}
If, in addition, $\Psi(a,b) = \Psi(b,a)$ for all $a,b \in A$, then $\Psi$ is called \emph{symmetric}.  Note that Delsarte includes symmetry as part of the definition of an inner product; we do not.  Inner products, but with values in $\Q/\Z$ instead of $\C^\times$, also figure prominently in \cite{nebe-rains-sloane, wood:turkey}.

\begin{rem}
When $n \in \Z$, note that $\Psi(na,b) = (\Psi(a,b))^n$.
\end{rem}

Given a duality $\phi$ of a finite abelian group $A$, define $\Phi: A \times A \ra \C^\times$: 
\begin{equation}  \label{eqn:DualityToInnerProduct}
\Phi(a,b) = \pair{\phi(a)}{b}, \quad a,b \in A .
\end{equation}
Conversely, given an inner product $\Psi: A \times A \ra \C^\times$, define $\psi: A \ra \ah$:
\begin{equation}  \label{eqn:InnerProductToDuality}
\pair{\psi(a)}{b} = \Psi(a,b), \quad a, b \in A .
\end{equation}

\begin{prop}
If $\phi: A \ra \ah$ is a duality of $A$, then $\Phi$ of \eqref{eqn:DualityToInnerProduct} is an inner product on $A$.  Conversely, if $\Psi: A \times A \ra \C^\times$ is an inner product on $A$, then $\psi$ of \eqref{eqn:InnerProductToDuality} is a duality of $A$.  Moreover, for any duality $\phi: A \ra \ah$ of $A$, the inner product $\Phi^*$ associated to the duality $\phi^*: A \ra \ah$ satisfies $\Phi^*(a,b) = \Phi(b,a)$ for all $a,b \in A$.  In particular, a duality $\phi$ is symmetric if and only if its associated inner product $\Phi$ is symmetric.
\end{prop}

\begin{proof}
The first two properties for $\Phi$ to be an inner product follow from \eqref{eqn:CharHom}, \eqref{eqn:CharProduct}, and $\phi$ being a homomorphism.  The third property holds because $\phi$ is injective, and the fourth property holds, via Proposition~\ref{prop:SizeCharGroup}, because $\phi$ is surjective.  Essentially the same arguments yield $\psi$ being a duality.
The relationship between $\Phi^*$ and $\Phi$, as well as the statement about symmetry, follow from Lemma~\ref{lem:SymDual}.  
\end{proof}

\begin{ex}  \label{ex:Klein-4-group}
Let $A = \F_2^2$, the Klein $4$-group.  Write elements of $A$ as pairs $ab$, with $a, b \in \F_2$.  Then $\aut(A) = \GL(2, \F_2)$, which is isomorphic to the dihedral group $D_3$ of order $6$ (also isomorphic to the symmetric group of degree $3$).  The automorphisms permute the three nonzero vectors in $\F_2^2$.  

Define a symmetric duality $\phi_0$ of $A$ by
\[  \Phi_0(ab,cd) = \pair{\phi_0(ab)}{cd} = (-1)^{ac+bd} = (-1)^{ab [cd]^\top}  . \]
The characters $\pi_i=\phi_0(ab)$ have the following values:
\[  \begin{array}{cccccc}
\pi & ab & \pair{\pi}{00} & \pair{\pi}{01} & \pair{\pi}{10} & \pair{\pi}{11} \\ \hline
\pi_0 & 00 & 1 & \hphantom{-}1 & \hphantom{-}1 & \hphantom{-}1  \\
\pi_1 & 01 & 1 & -1 & \hphantom{-}1 & -1  \\
\pi_2 & 10 & 1 & \hphantom{-}1 & -1 & -1  \\
\pi_3 & 11 & 1 & -1 & -1 & \hphantom{-}1 
\end{array} \]

For the six elements $P \in \aut(A)$, here are the associated dualities $\phi_P(ab) = \phi_0(ab P)$, together with $\phi^*_P$ and the group order $o(P)$ of $P$.
\[  \begin{array}{cccccccc}
\phi_i & P & \phi_P(00) & \phi_P(01) & \phi_P(10) & \phi_P(11) & \phi^*_P & o(P) \\ \hline
\phi_0 & \left[ \begin{smallmatrix}
1 & 0 \\ 0 & 1
\end{smallmatrix} \right] & \pi_0 & \pi_1 & \pi_2 & \pi_3 & \phi_0 & 1 \\
\phi_1 & \left[ \begin{smallmatrix}
1 & 1 \\ 1 & 0 
\end{smallmatrix} \right] & \pi_0 & \pi_2 & \pi_3 & \pi_1 & \phi_1 & 3 \\
\phi_2 & \left[ \begin{smallmatrix}
0 & 1 \\ 1 & 1
\end{smallmatrix} \right] & \pi_0 & \pi_3 & \pi_1 & \pi_2 & \phi_2 & 3 \\
\phi_3 & \left[ \begin{smallmatrix}
0 & 1 \\ 1 & 0
\end{smallmatrix} \right] & \pi_0 & \pi_2 & \pi_1 & \pi_3 & \phi_3 & 2 \\ \hline
\phi_4 & \left[ \begin{smallmatrix}
1 & 1 \\ 0 & 1
\end{smallmatrix} \right] & \pi_0 & \pi_1 & \pi_3 & \pi_2 & \phi_5 & 2 \\
\phi_5 & \left[ \begin{smallmatrix}
1 & 0 \\ 1 & 1
\end{smallmatrix} \right] & \pi_0 & \pi_3 & \pi_2 & \pi_1 & \phi_4 & 2  
\end{array}  \]

Four of the six dualities are symmetric: $\phi_i$, $i=0,1,2,3$.  The remaining two dualities, $\phi_4, \phi_5$, form a nonsymmetric pair: $\phi_4^* = \phi_5$ and $\phi^*_5 = \phi_4$.
The same dualities are listed in \cite[Example~2]{MR4770737}, but symmetry there (also in \cite[Theorem~2.5]{STD-DualitiesFinAbGp}) is tied to group order, which is contrary to the table above.
The table shows that the order of an automorphism does not determine whether its corresponding duality is symmetric.
\end{ex}

\begin{ex}  \label{ex:2by4}
Let $A = \Z/2\Z \times \Z/4\Z$.  Write elements of $A$ as a pair $ab$ with $a \in \Z/2\Z$ and $b \in \Z/4\Z$.  The elements $01, 03, 11, 13$ of $A$ have order $4$, while the elements $02, 10, 12$ have order $2$.  One set of generators of the group $A$ is $\{10 , 01\}$.  Any character of $A$ is determined by its values on the generators.

Define a symmetric duality $\phi_0$ of $A$ by 
\[  \Phi_0(ab,cd) = \pair{\phi_0(ab)}{cd} = (-1)^{ac} i^{bd} .  \]
The characters $\pi_i = \phi_0(ab)$ of $A$ are listed next, $ab$ vertically, $cd$ horizontally, with entries equal to $\Phi_0(ab,cd) = \pair{\phi_0(ab)}{cd}$.
\[  \begin{array}{cr|rrrrrrrr}
&  & 00 & 01 & 02 & 03 & 10 & 11 & 12 & 13 \\ \hline
\pi_0 & 00 &   1 & 1 & 1  & 1 & 1 & 1 & 1 & 1 \\
\pi_1 & 01 & 1 & i & -1 & -i & 1 & i & -1 & -i  \\
\pi_2 & 02 & 1 & -1 & 1 & -1 & 1 & -1 & 1 & -1 \\
\pi_3 & 03 & 1 & -i & -1 & i & 1 & -i & -1 & i \\
\pi_4 & 10 & 1 & 1 & 1 & 1 & -1 & -1 & -1 & -1 \\
\pi_5 & 11 & 1 & i & -1 & -i & -1 & -i & 1 & i  \\
\pi_6 & 12 & 1 & -1 & 1 & -1 & -1 & 1 & -1 & 1 \\
\pi_7 & 13 & 1 & -i & -1 & i & -1 & i & 1 & -i  \\
\end{array}
\]

As with characters, an automorphism of $A$ is completely determined by its values on the generators.  An automorphism must send $01$ to one of the elements of order $4$ and $10$ to either $10$ or $12$ (not to $02$, which is twice each of the elements of order $4$).  Write each automorphism as a $2 \times 2$ matrix, with first row equal to the image of $10$ and second row equal to the image of $01$.  Setting
\[  \sigma = \begin{bmatrix}
1 & 0 \\ 1 & 1
\end{bmatrix} , \quad 
\tau = \begin{bmatrix}
1 & 2 \\ 1 & 1
\end{bmatrix} , \]
one recognizes $\aut(A)$ to be the dihedral group $D_4$ of order $8$, with $\sigma^2=I$, $\tau^4=I$, and $\tau \sigma = \sigma \tau^3$.

For the eight elements $P \in \aut(A)$, here are the associated dualities $\phi_P(ab) = \phi_0(abP)$.
\[  \begin{array}{cccccc}
& \sigma^\epsilon \tau^j & P & \phi_P(01) & \phi_P(10) & \phi_P^* \\ \hline
\phi_0 & I & \left[ \begin{smallmatrix}
1 & 0 \\ 0 & 1
\end{smallmatrix} \right] & \pi_1 & \pi_4 & \phi_0 \\
\phi_1 & \tau & \left[ \begin{smallmatrix}
1 & 2 \\ 1 & 1
\end{smallmatrix} \right] & \pi_5 & \pi_6 & \phi_1 \\
\phi_2 & \tau^2 & \left[ \begin{smallmatrix}
1 & 0 \\ 0 & 3
\end{smallmatrix} \right] & \pi_3 & \pi_4 & \phi_2 \\
\phi_3 & \tau^3 & \left[ \begin{smallmatrix}
1 & 2 \\ 1 & 3
\end{smallmatrix} \right] & \pi_7 & \pi_6 & \phi_3 \\ \hline
\phi_4 & \sigma & \left[ \begin{smallmatrix}
1 & 0 \\ 1 & 1
\end{smallmatrix} \right] & \pi_5 & \pi_4 & \phi_7 \\
\phi_5 & \sigma \tau & \left[ \begin{smallmatrix}
1 & 2 \\ 0 & 3
\end{smallmatrix} \right] & \pi_3 & \pi_6 & \phi_6 \\
\phi_6 & \sigma \tau^2 & \left[ \begin{smallmatrix}
1 & 0 \\ 1 & 3
\end{smallmatrix} \right] & \pi_7 & \pi_4 & \phi_5 \\
\phi_7 & \sigma \tau^3 & \left[ \begin{smallmatrix}
1 & 2 \\ 0 & 1
\end{smallmatrix} \right] & \pi_1 & \pi_6 & \phi_4 
\end{array}  \]
Four of the dualities are symmetric: $\phi_i$, $i=0,1,2,3$.  The other dualities form two pairs: $\phi^*_4 = \phi_7$ and $\phi_5^* = \phi_6$, contrary to \cite[Example~4]{STD-DualitiesFinAbGp}.
\end{ex}

\section{Additive codes and dual codes}  \label{sec:AddDualCodes}

In this section, we define additive codes and use a choice of duality to define dual codes.

Let $A$ be a finite abelian group, and choose a duality $\phi: A \ra \ah$.  Using Lemma~\ref{lem:DualProduct}, $\phi$ induces a duality $A^n \ra \ah^n$ by setting
\begin{equation}  \label{eqn:DiagExtension}
\pair{\phi(a_1, a_2, \ldots, a_n)}{(b_1, b_2, \ldots, b_n)}_{A^n} = \prod_{i=1}^n \pair{\phi(a_i)}{b_i}_A ,
\end{equation}
 for $(a_1, a_2, \ldots, a_n), (b_1, b_2, \ldots, b_n) \in A^n$.  If $\phi: A \ra \ah$ is symmetric, so is its extension $\phi: A^n \ra \ah^n$.
Extend the inner product $\Phi$ on $A$ to an inner product on $A^n$ (still called $\Phi$, abusing notation) by
\[  \Phi((a_1, a_2, \ldots, a_n),(b_1, b_2, \ldots, b_n))= \prod_{i=1}^n \Phi(a_i,b_i) , \]
for $(a_1, a_2, \ldots, a_n),(b_1, b_2, \ldots, b_n) \in A^n$.

\begin{rem}
Not every duality of $A^n$ has the form of \eqref{eqn:DiagExtension}.  In Example~\ref{ex:Klein-4-group}, the duality $\phi_3$ of $\F_2^2$ is not equal to $\phi_1$.  This is even true up to the appropriate notion of equivalence, as will be addressed in Section~\ref{sec:congruence}.
\end{rem}

An \emph{additive code} of length $n$ over $A$ is a subgroup $C \subseteq A^n$.  An additive code has an annihilator $(\ah^n: C)$, as in \eqref{eqn:DefnAnnihilator}.
The annihilator $(\ah^n: C)$ has most of the properties one would want in a dual code, including the size condition $\size{C} \cdot \size{(\ah^n: C)} = \size{A^n}$ and the double annihilator peoperty $(A^n : (\ah^n:C)) = C$, Corollary~\ref{cor:DoubleAnnihilator}; cf., \cite[\S 11.2]{wood:turkey}.  The only drawback is that the annihilator $(\ah^n: C)$ is contained in $\ah^n$, not in $A^n$.  The entire reason for discussing dualities is to be able to pull back the annihilator $(\ah^n: C)$ to live in $A^n$.

For an additive code $C \subseteq A^n$ and a choice of duality $\phi: A \ra \ah$, define left and right \emph{dual codes} by 
\begin{align*}
\mathfrak{L}_{\phi}(C) = \{ x \in A^n: \Phi(x,c)=1 \text{ for all $c \in C$} \}, \\
\mathfrak{R}_{\phi}(C) = \{ x \in A^n: \Phi(c,x)=1 \text{ for all $c \in C$} \} .
\end{align*}
We may write $\mathfrak{L}(C)$ or $\mathfrak{R}(C)$ when $\phi$ is unambiguous.

\begin{lem}  \label{lem:LeftRightDuals}
Given a finite abelian group $A$ and a duality $\phi: A \ra \ah$, the following hold for all additive codes $C, C_1, C_2 \subseteq A^n$:
\begin{itemize}
\item $\phi(\mathfrak{L}_{\phi}(C)) = (\ah^n:C)$ and $\phi^*(\mathfrak{R}_{\phi}(C)) = (\ah^n:C)$.
\item  If $C_1 \subseteq C_2 \subseteq A^n$, then $\mathfrak{L}_{\phi}(C_2) \subseteq \mathfrak{L}_{\phi}(C_1)$ and $\mathfrak{R}_{\phi}(C_2) \subseteq \mathfrak{R}_{\phi}(C_1)$.
\item  $\mathfrak{L}_{\phi^*}(C)=\mathfrak{R}_{\phi}(C)$ and $\mathfrak{R}_{\phi^*}(C)= \mathfrak{L}_{\phi}(C)$.
\item  If the duality $\phi$ is symmetric, then $\mathfrak{L}_{\phi}(C)=\mathfrak{R}_{\phi}(C)$.
\end{itemize}
\end{lem}

\begin{proof}
These are exercises using Lemma~\ref{lem:SymDual}.
\end{proof}

The left dual code $\mathfrak{L}_{\phi}(C)$ corresponds to the orthogonal $C^M$ of \cite[Definition~2.2]{STD-DualitiesFinAbGp}, and $\mathfrak{R}_{\phi}(C)$ corresponds to $C^{M^\top}$.  

\begin{prop}  \label{prop:DualProperties}
Given a finite abelian group $A$ and a duality $\phi: A \ra \ah$, the dual codes of any additive code $C \subseteq A^n$ have the following properties:
\begin{itemize}
\item  $\mathfrak{L}_{\phi}(C)$ and $\mathfrak{R}_{\phi}(C)$ are additive codes in $A^n$.
\item  $\size{\mathfrak{L}_{\phi}(C)} \cdot \size{C} = \size{A}^n$ and $\size{\mathfrak{R}_{\phi}(C)} \cdot \size{C} = \size{A}^n$.
\item  $\mathfrak{L}_{\phi}(\mathfrak{R}_{\phi}(C)) = C$ and $\mathfrak{R}_{\phi}(\mathfrak{L}_{\phi}(C))=C$.
\end{itemize}
\end{prop}

\begin{proof}
One verifies the first two items using Lemma~\ref{lem:LeftRightDuals} and Corollary~\ref{cor:DoubleAnnihilator}.  For the last item, first show that $C$ is contained in the double dual, and then use the size condition to prove equality.
\end{proof}

The next proposition is a version of Proposition~\ref{prop:VanishingSums}.
\begin{prop}  \label{prop:VanishingSumsForPhi}
Let $\phi: A \ra \ah$ be a duality of $A$, extended to $A^n$, with associated inner product $\Phi$.  For any additive code $C \subseteq A^n$,
\[  \sum_{y \in C} \Phi(x,y) = \begin{cases}
\size{C}, & x \in \mathfrak{L}_\phi(C) , \\
0, & x \not\in \mathfrak{L}_\phi(C) ;
\end{cases} 
\quad
\sum_{x \in C} \Phi(x,y) = \begin{cases}
\size{C}, & y \in \mathfrak{R}_\phi(C) , \\
0, & y \not\in \mathfrak{R}_\phi(C) .
\end{cases} \]
\end{prop}

\begin{proof}
In the first case, $\sum_{y \in C} \Phi(x,y) = \sum_{y \in C} \pair{\phi(x)}{y}$ for $x \in A^n$.  Using that $x \in \mathfrak{L}_\phi(C)$ if and only if $\phi(x) \in (\ah^n: C)$, the result follows from Proposition~\ref{prop:VanishingSums}.  The second case follows from applying the first case to the duality $\phi^*$.
\end{proof}

There are versions of the MacWilliams identities that hold using these dual codes.  This will be the topic of Section~\ref{sec:MWIds}.

Because $\mathfrak{L}_{\phi}(C), \mathfrak{R}_{\phi}(C) \subseteq A^n$, it is possible to define self-orthogonal and self-dual codes (with left-right modifiers):
\begin{itemize}
\item left self-orthogonal: $C \subseteq \mathfrak{L}_{\phi}(C)$;
\item right self-orthogonal: $C \subseteq \mathfrak{R}_{\phi}(C)$;
\item left self-dual: $C = \mathfrak{L}_{\phi}(C)$;
\item right self-dual: $C = \mathfrak{R}_{\phi}(C)$.
\end{itemize}
In fact, the left-right distinction is not needed, as the next result shows.

\begin{lem}
An additive code $C \subseteq A^n$ is left self-orthogonal if and only if $C$ is right self-orthogonal.  Similarly, $C$ is left self-dual if and only if $C$ is right self-dual.
\end{lem}

\begin{proof}
Suppose $C \subseteq \mathfrak{L}(C)$.  Take the right dual of both sides and use Lemma~\ref{lem:LeftRightDuals} and Proposition~\ref{prop:DualProperties}.  Then $C = \mathfrak{R}(\mathfrak{L}(C)) \subseteq \mathfrak{R}(C)$.  The other proofs are similar.
\end{proof}

\begin{ex}  \label{ex:Klein4}
Let $A$ be the Klein $4$-group, viewed as row vectors $a = [ a_1, a_2 ]$ over the binary field $\F_2$. 
The six dualities of $A$ appear in Example~\ref{ex:Klein-4-group}, each having the form
\[  \Phi_P(a,b)= \pair{\phi_P(a)}{b} = (-1)^{a P b^\top} \in \C^\times, \quad a, b \in A, \]
for $P \in \GL(2,\F_2)$.

There are three subgroups of $A$ of order two.  (We will write elements without brackets.)  The subgroups are:
\[ C_0 = \{ 00, 10 \}, \quad C_1 = \{ 00, 11\}, \quad C_{\infty} = \{ 00, 01 \} . \]
The dual codes of these three subgroups will also have order two.  For each matrix $P \in \GL(2,\F_2)$, here are the left and right dual codes.
\[  \begin{array}{c|cc|cc|cc}
P & \mathfrak{L}(C_0) & \mathfrak{R}(C_0) & \mathfrak{L}(C_1) & \mathfrak{R}(C_1) & \mathfrak{L}(C_\infty) & \mathfrak{R}(C_\infty) \\ \hline 
\begin{bsmallmatrix}
1 & 0 \\ 0 & 1
\end{bsmallmatrix} & C_\infty & C_\infty & C_1 & C_1 & C_0 & C_0 \\
\begin{bsmallmatrix}
0 & 1 \\ 1 & 1
\end{bsmallmatrix} & C_0 & C_0 & C_\infty & C_\infty & C_1 & C_1 \\
\begin{bsmallmatrix}
1 & 1 \\ 1 & 0
\end{bsmallmatrix} & C_1 & C_1 & C_0 & C_0 & C_\infty & C_\infty \\
\begin{bsmallmatrix}
0 & 1 \\ 1 & 0
\end{bsmallmatrix} & C_0 & C_0 & C_1 & C_1 & C_\infty & C_\infty \\ \hline
\begin{bsmallmatrix}
1 & 1 \\ 0 & 1
\end{bsmallmatrix} & C_\infty & C_1 & C_0 & C_\infty & C_1 & C_0 \\
\begin{bsmallmatrix}
1 & 0 \\ 1 & 1
\end{bsmallmatrix} & C_1 & C_\infty & C_\infty & C_0 & C_0 & C_1 
\end{array}  \]

For each of the first three matrices $P$, there is exactly one self-dual code (with a different self-dual code for each $P$).  For $P = \begin{bsmallmatrix}
0 & 1 \\ 1 & 0
\end{bsmallmatrix}$, all three codes are self-dual.  For the two matrices $P$ that are not symmetric, there are no self-dual codes and the left/right dual codes are different.  We will come back to the self-dual codes in Example~\ref{ex:Klein4-bis}.
\end{ex}

\begin{ex}
Let $A = \F_2^3$.  There are $\size{\GL(3,\F_2)} = 168$ dualities, of which $28$ (one-sixth of the total) are symmetric, \eqref{eqn:SymInvertCount}.

Pick $P \in \GL(3, \F_2)$ that is not symmetric, say
\[  P = \begin{bmatrix}
0 & 0 & 1 \\
1 & 1 & 0 \\
1 & 0 & 0
\end{bmatrix}  . \]

Let $C = \{ 000, 100 \}$. Then $\mathfrak{L}(C) = \{ 000, 100, 011, 111\}$, while $\mathfrak{R}(C) = \{ 000, 100, 010, 110 \}$.  We have $C = \mathfrak{L}(C) \cap \mathfrak{R}(C)$, but $\mathfrak{L}(C) \neq \mathfrak{R}(C)$.  The code $C$ is left and right self-orthogonal, but the left/right dual codes are different.
\end{ex}

We know from Proposition~\ref{prop:DualProperties} that for subgroups $H,K \subseteq A$, if $K = \mathfrak{L}_\phi(H)$ for some duality $\phi$ of $A$, then $\size{H} \cdot \size{K} = \size{A}$.  The converse was addressed, for elementary abelian $2$-groups, in \cite[Theorem~15]{MR4770737}, 
and, for arbitrary finite abelian groups, in \cite[Theorem~5]{MR4332568}.  The statement of the latter result turns out to be too optimistic, as will be seen in the next several results.

\begin{prop}  \label{prop:direct-sum-duals} 
Let $A$ be a finite abelian group with subgroups $H,K \subseteq A$ such that $A = H \oplus K$.  Then there exists a symmetric duality $\phi: A \ra \ah$ such that $K = \mathfrak{L}(H) =  \mathfrak{R}(H)$ and $H = \mathfrak{L}(K) = \mathfrak{R}(K)$.
\end{prop}

\begin{proof}
The direct sum hypothesis implies $\size{A} = \size{H} \cdot \size{K}$.  Write elements of $A = H \oplus K$ as pairs $(h,k)$ with $h \in H$ and $k \in K$.  

Let $\phi_H: H \ra \charac{H}$ and $\phi_K: K \ra \charac{K}$ be symmetric dualities of $H$ and $K$.  Define $\phi: A \ra \ah$ to be $\phi_H \times \phi_K$, Lemma~\ref{lem:DualProduct}.  That is,
\[  \pair{\phi(h,k)}{(h',k')}_A = \pair{\phi_H(h)}{h'}_H \, \pair{\phi_K(k)}{k'}_K \in \C^\times . \]
Then direct calculation yields 
\begin{align*}
\pair{\phi(0,k)}{(h',0)}_A &= \pair{\phi_H(0)}{h'}_H \, \pair{\phi_K(k)}{0}_K = 1 , \\
\pair{\phi(h,0)}{(0,k')}_A &= \pair{\phi_H(h)}{0}_H \, \pair{\phi_K(0)}{k'}_K = 1 ,
\end{align*}
so that $K \subseteq \mathfrak{L}(H)$ and $K \subseteq \mathfrak{R}(H)$ as well as $H \subseteq \mathfrak{L}(K)$ and $H \subseteq \mathfrak{R}(K)$.  Equality follows by the size condition, Proposition~\ref{prop:DualProperties}.
\end{proof}

There are two situations where Proposition~\ref{prop:direct-sum-duals} can be generalized to any two subgroups satisfying the size condition: cyclic $p$-groups and 
elementary abelian $p$-groups.  

\begin{prop}
Let $A = \Z/p^k\Z$, for some prime $p$.  Suppose $H, K \subseteq A$ are subgroups of $A$ that satisfy $\size{H} \cdot \size{K} = \size{A}$.  Then $\mathfrak{L}(H) = \mathfrak{R}(H) = K$ and $\mathfrak{L}(K) = \mathfrak{R}(K) = H$ for every duality $\phi$ of $A$.
\end{prop}

\begin{proof}
The group $A$ is very special: for any $j=0, 1, \ldots, k$, there is a unique subgroup $A_j$ of order $p^j$.  Subgroups that satisfy $\size{H} \cdot \size{K} = \size{A}$ are of the form $H=A_j$ and $K = A_{k-j}$ for some $j=0, 1, \ldots, k$.  The size condition for dual codes, Proposition~\ref{prop:DualProperties}, forces $A_j$ and $A_{k-j}$ to be dual codes for any duality.
\end{proof}

\begin{thm}  \label{thm:SizeDuals}
Let $A$ be an elementary abelian $p$-group.  Suppose $H, K \subseteq A$ are subgroups of $A$ such that $\size{H} \cdot \size{K} = \size{A}$.  Then there exists a symmetric duality $\phi$ of $A$ such that $\mathfrak{L}(H) = \mathfrak{R}(H) = K$ and $\mathfrak{L}(K) = \mathfrak{R}(K) = H$.
\end{thm}

\begin{proof}
View $A$ as $\F_p^n$ and $H, K$ as linear subspaces.  Write $h = \dim H$ and $k = \dim K$.  The cardinality hypothesis says that $h + k = n$.

If $i=\dim(H \cap K) > 0$, then choose a basis $e_1, e_2, \ldots, e_i$ of $H \cap K$.  (If $i=0$, the basis of $H \cap K$ is empty.)  Choose elements $e_{i+1}, \ldots, e_h$ so that $e_1, \ldots, e_h$ is a basis of $H$.  Choose $e_{h+1}, \ldots, e_{h+k-i}$ so that $e_1, \ldots, e_i, e_{h+1}, \ldots, e_{h+k-i}$ is a basis of $K$.  Then $e_1, \ldots, e_{h+k-i}$ is a basis of $H+K$.  Choose $e_{h+k-i+1}, \ldots, e_n$, so that $e_1, \ldots, e_n$ is a basis of $A$.

Form the dual basis $\pi_1, \ldots, \pi_n$ of $\ah$ with the property that 
\[  \pair{\pi_j}{e_\ell} = \begin{cases}
\zeta_p, & j = \ell, \\
1, & j \neq \ell ,
\end{cases} \]
for all $j ,\ell = 1, 2, \ldots, n$, where $\zeta_p$ is a primitive $p$th root of $1$ in $\C^\times$.

We define $\phi: A \ra \ah$ by specifying the values of $\phi$ on the basis $e_1, \ldots, e_n$ of $A$.  For  convenience, set $c = h + k -i$.  Note that $c+i = h + k = n$.  Define
\[  \phi(e_j) = \begin{cases}
\pi_{c+j}, & j = 1, 2, \ldots, i, \\
\pi_j, & j= i+1, i+2, \ldots, c, \\
\pi_{j-c}, & j=c+1, c+2, \ldots, c+i .
\end{cases}  \]
This $\phi$ takes a basis of $A$ to a basis of $\ah$, so $\phi$ defines a duality of $A$.  By examining cases, one verifies that $\phi$ is symmetric and, using the size condition of Proposition~\ref{prop:DualProperties}, that $H$ and $K$ are duals of each other.
\end{proof}

The next example shows that Theorem~\ref{thm:SizeDuals} does not generalize further, contrary to \cite[Theorem~5]{MR4332568}.

\begin{ex}  \label{ex:4+2}
Let $A = \Z/2\Z \times \Z/4\Z$, so that $\size{A} = 8$.  Example~\ref{ex:2by4} displays the dualities of $A$.  Here, we determine the dual codes of the subgroups of $A$ with respect to those dualities.

There are three subgroups of $A$ having order $2$: $\ell_0 = \{00, 10\}$, $\ell_1 = \{00, 12\}$, and $\ell_\infty = \{00, 02\}$.  There are also three subgroups of order $4$: $C_1 = \{00, 01, 02, 03\}$, $C_2 = \{00, 11, 02, 13\}$, and $S = \{00,10,02,12\}$; $C_1, C_2$ are cyclic groups, while $S$, the socle of $A$, is elementary abelian.

The following table displays the left and right dual codes of the subgroups of order $2$ with respect to the various dualities.

\[  \begin{array}{c|cccccc}
\phi & \mathfrak{L}(\ell_0) & \mathfrak{R}(\ell_0) & \mathfrak{L}(\ell_1) & \mathfrak{R}(\ell_1) & \mathfrak{L}(\ell_\infty) & \mathfrak{R}(\ell_\infty) \\ \hline
\phi_0 & C_1 & C_1 & C_2 & C_2 & S & S  \\
\phi_1 & C_2 & C_2 & C_1 & C_1 & S & S  \\
\phi_2 & C_1 & C_1 & C_2 & C_2 & S & S  \\
\phi_3 & C_2 & C_2 & C_1 & C_1 & S & S  \\
\phi_4 & C_2 & C_1 & C_1 & C_2 & S & S   \\
\phi_5 & C_1 & C_2 & C_2 & C_1 & S & S   \\
\phi_6 & C_2 & C_1 & C_1 & C_2 & S & S   \\
\phi_7 & C_1 & C_2 & C_2 & C_1 & S & S  
\end{array}  \]
By using double duals, Proposition~\ref{prop:DualProperties}, one can determine the dual codes of $C_1, C_2, S$.

Note that $\size{\ell_\infty} \cdot \size{C_1} = \size{A}$, but there is no duality with $\mathfrak{L}(\ell_\infty) = C_1$, contrary to \cite[Theorem~5]{MR4332568}.
\end{ex}

This example will be generalized in Theorem~\ref{thm:KerImDuals}; cf., Remark~\ref{rem:KerImExample}.

\section{Structural questions}  \label{sec:StructureQ}
In this section, we study the problem of understanding how the dual codes of a subgroup $H \subseteq A$ depend on the choice of duality.  On one extreme, there are elementary abelian $p$-groups, where Theorem~\ref{thm:SizeDuals} says that any two subgroups satisfying the size condition are dual codes under \emph{some} duality.  On the other extreme is Example~\ref{ex:4+2}, which provides examples of subgroups of $A = \Z/2\Z \times \Z/4 \Z$ that are dual codes for \emph{every} duality.  We will find that the dependence of the dual codes on the duality is intimately related to the action on subgroups of the group of automorphisms. 

Let $A$ be a finite abelian group.  
The automorphism group $\aut(A)$ acts on $A$.  We will write this action as a right action, with inputs written on the left.  
Let $\mathscr{S}_d$ be the set of all subgroups of $A$ having order $d$.  Then $\aut(A)$ also acts on $\mathscr{S}_d$ on the right.    For a subgroup $H \subseteq A$ with $\size{H}=d$, i.e., $H \in \mathscr{S}_d$, let $\stab(H)$ be its stabilizer subgroup:
\[  \stab(H) = \{ \tau \in \aut(A): H \tau = H \} . \]

\begin{lem}  \label{lem:stab-mults}
Let $A$ be a finite abelian group.  Take any subgroup $H \subseteq A$, any automorphism $\tau \in \aut(A)$, and any duality $\phi$ of $A$.  Then,  $\mathfrak{R}_{\phi}(H \tau) = \mathfrak{R}_{\phi}(H)$ if and only if $\tau \in \stab(H)$.  Likewise, $\mathfrak{L}_{\phi}(H \tau) = \mathfrak{L}_{\phi}(H)$ if and only if $\tau \in \stab(H)$.
\end{lem}

\begin{proof}
By the double dual property, $\mathfrak{R}_{\phi}(H \tau) = \mathfrak{R}_{\phi}(H)$ if and only if $H \tau = H$.  The same reasoning applies to left dual codes.
\end{proof}

\begin{lem}  \label{lem:dual-of-aut-mult}
Let $A$ be a finite abelian group, with subgroup $H \subseteq A$.
Suppose dualities $\phi_1, \phi_2$ of $A$ satisfy $\phi_2 = \phi_1 \circ \tau$ for some $\tau \in \aut(A)$.  Then $\mathfrak{R}_{\phi_2}(H) = \mathfrak{R}_{\phi_1}(H \tau)$ and $\mathfrak{L}_{\phi_2}(H) \tau = \mathfrak{L}_{\phi_1}(H)$.
\end{lem}

\begin{proof}
Calculate:
\begin{align*}
\mathfrak{R}_{\phi_2}(H) &= \{ y \in A: \pair{\phi_2(h)}{y} = 0 \text{ for all $h \in H$} \} \\
&= \{ y \in A: \pair{\phi_1(h \tau)}{y} = 0 \text{ for all $h \in H$} \} 
= \mathfrak{R}_{\phi_1}(H \tau) ;  \\
\mathfrak{L}_{\phi_2}(H) &= \{ x \in A: \pair{\phi_2(x)}{h}=0 \text{ for all $h \in H$}\} \\
&= \{ x \in A: \pair{\phi_1(x \tau)}{h}=0 \text{ for all $h \in H$}\} 
=  \mathfrak{L}_{\phi_1}(H) \tau^{-1}. \qedhere
\end{align*}
\end{proof} 

\begin{prop}  \label{prop:SameDuals-stab}
Let $H$ be a subgroup of a finite abelian group $A$.  Suppose $\phi_1, \phi_2$ are two dualities of $A$.  Then $\mathfrak{R}_{\phi_1}(H) = \mathfrak{R}_{\phi_2}(H)$ if and only if $\phi_2 = \phi_1 \circ \tau$ for some $\tau \in \stab(H)$.  Likewise, $\mathfrak{L}_{\phi_1}(H) = \mathfrak{L}_{\phi_2}(H)$ if and only if $\phi_2^* = \phi_1^* \circ \tau$ for some $\tau \in \stab(H)$. 
\end{prop}

\begin{proof}
By Proposition~\ref{prop:DualitiesAut}, $\phi_2 = \phi_1 \circ \tau$ for some $\tau \in \aut(A)$.  Then,
$\mathfrak{R}_{\phi_2}(H) = \mathfrak{R}_{\phi_1}(H \tau)$, by Lemma~\ref{lem:dual-of-aut-mult}.
Thus, by Lemma~\ref{lem:stab-mults}, $\mathfrak{R}_{\phi_1}(H) = \mathfrak{R}_{\phi_2}(H)$ if and only if $\tau \in \stab(H)$.  For left duals, apply the right dual case to $\phi_1^*$ and $\phi_2^*$, using Lemma~\ref{lem:LeftRightDuals}.
\end{proof}

Recall that $H \subseteq A$ a \emph{characteristic subgroup} if $H$ is invariant under every automorphism of $A$, i.e., $H \tau = H$ for every $\tau \in \aut(A)$, or, equivalently, $\stab(H) = \aut(A)$.

\begin{thm}  \label{thm:char-subgroups}
Let $H$ and $K$ be subgroups of a finite abelian group $A$.  Suppose that $K = \mathfrak{R}_{\phi_0}(H)$ for some duality $\phi_0$ of $A$.  Then $K = \mathfrak{R}_{\phi}(H)$ for every duality $\phi$ of $A$ if and only if $H$ is  a characteristic subgroup.  Likewise for left dual codes.  Moreover, $K = \mathfrak{R}_{\phi}(H)$ for every duality $\phi$ of $A$ if and only if $K = \mathfrak{L}_{\phi}(H)$ for every duality $\phi$ of $A$.
\end{thm}

\begin{proof}
Use Proposition~\ref{prop:SameDuals-stab}.
\end{proof}

\begin{cor}
Suppose $H, K$ are subgroups of a finite abelian group $A$, with $K = \mathfrak{R}_{\phi_0}(H)$ for some duality $\phi_0$ of $A$.  Then, $H$ is a characteristic subgroup if and only if $K$ is a characteristic subgroup.
\end{cor}

\begin{proof}
If $H$ is a characteristic subgroup, then, by Theorem~\ref{thm:char-subgroups}, $K = \mathfrak{R}_\phi(H)$ for any duality $\phi$ of $A$.  Since $\mathfrak{R}_\phi(H) = \mathfrak{L}_{\phi^*}(H)$, Lemma~\ref{lem:LeftRightDuals}, we also have $K = \mathfrak{L}_\phi(H)$ for any duality $\phi$ of $A$.

Take any automorphism $\tau \in \aut(A)$.  Set $\phi = \phi_0 \circ \tau$.  By Lemma~\ref{lem:dual-of-aut-mult}, we know that $\mathfrak{L}_{\phi}(H) \tau = \mathfrak{L}_{\phi_0}(H)$.  But that means $K \tau = K$, and $K$ is a characteristic subgroup.

Essentially the same argument applies when $K$ is a characteristic subgroup, with $H = \mathfrak{L}_{\phi_0}(K)$.
\end{proof}

Proposition~\ref{prop:PrimeByPrime} allows us to study finite abelian groups one prime at a time.  So, for the rest of this section, we assume $A$ is a finite abelian $p$-group for some fixed prime $p$.  We will present two related filtrations of $A$.

Define $f: A \ra A$ by $f(a) = pa$, $a \in A$; $f$ is a group homomorphism.  Denote composition of $f$ with itself using exponents, so that $f^2 = f \circ f$.  Then $f^k(a) = p^k a$, $a \in A$, $k$ positive integer.  We use the convention that $f^0 = \id_A$.   Because $A$ is a finite abelian $p$-group, there exists a smallest positive integer $N$ such that $f^N = 0$.  (By the fundamental theorem of finite abelian groups, $A$ is a direct sum of cyclic groups whose orders are powers of $p$.  If $p^N$ is the largest power that appears, then $f^N = 0$.)

We have the following filtrations:
\begin{align}
0 &= \ker f^0 \subseteq \ker f \subseteq \ker f^2 \subseteq \cdots \subseteq \ker f^{N-1} \subseteq \ker f^N = A ,  \label{eqn:TwoFiltrations} \\
A &= \im f^0 \supseteq \im f \supseteq \im f^2 \supseteq \cdots \supseteq \im f^{N-1} \supseteq \im f^N = 0 .  \notag
\end{align}

\begin{rem}
The filtrations in \eqref{eqn:TwoFiltrations} are examples of a socle series (for $\ker f^j$) and a radical or Loewy series (for $\im f^j$), viewing $A$ as a $\Z$-module, \cite[Definition~1.2.1]{MR1644252}.
\end{rem}

\begin{prop}  \label{prop:FiltrationCharacteristic}
Each subgroup in the filtrations \eqref{eqn:TwoFiltrations} is a characteristic subgroup of $A$.
\end{prop}

\begin{proof}
The homomorphism $f$ commutes with any automorphism $\tau$:  $(f(a))\tau = (pa)\tau = p(a \tau) = f(a \tau)$ for any $a \in A$.  This implies any $f^j$ commutes with any automorphism.  If $a \in\ker f^j$, then $f^j(a \tau) = (f^j(a))\tau = 0 \tau = 0$, so $a \tau \in \ker f^j$.  Argue similarly for $\im f^j$.
\end{proof}

\begin{thm}  \label{thm:KerImDuals}
Let $A$ be a finite abelian $p$-group with filtrations \eqref{eqn:TwoFiltrations}.  Then, for every $j=0, 1, \ldots, N$, and every duality $\phi: A \ra \ah$,
\begin{align*}
\im f^j &= \mathfrak{L}_{\phi}(\ker f^j) =  \mathfrak{R}_{\phi}(\ker f^j) , \\
\ker f^j &= \mathfrak{L}_{\phi}(\im f^j) =  \mathfrak{R}_{\phi}(\im f^j) .
\end{align*}
\end{thm}

\begin{proof}
By the fundamental theorem of finite abelian groups, $A$ can be written as a sum of cyclic $p$-groups:
\[  A = \bigoplus_{i=1}^\ell \Z/p^{n_i} \Z,  \]
for integers $1 \leq n_1 \leq \cdots \leq n_{\ell}$.  Write $a \in A$ in the corresponding form $a = (a_1, a_2, \ldots, a_\ell)$.  Fix $\zeta$ to be a primitive $p^{n_\ell}$th root of $1$ in $\C^\times$, and define a symmetric duality $\phi_0$ of $A$ by 
\[  \Phi_0(a,b) = \prod_{i=1}^\ell \zeta^{p^{n_\ell - n_i} a_i b_i} . \]

We show that $\ker f^j$ and $\im f^j$ are dual codes.  Let $a \in \ker f^j$ and $b \in \im f^j$, with $b = f^j(x) = p^j x$.  Then $a_i b_i =  a_i p^j x_i = 0$, for $i = 1, 2, \ldots, \ell$, because $a \in \ker f^j$.  Thus $\Phi_0(a,b)=1$, so that $\ker f^j  \subseteq \mathfrak{L}_{\phi_0}(\im f^j)$ and $\im f^j  \subseteq \mathfrak{R}_{\phi_0}(\ker f^j)$.  Equality holds in both cases because $\size{\ker f^j} \cdot \size{\im f^j} = \size{A}$ and the size condition for dual codes.  Because $\phi_0$ is symmetric, we also have $\ker f^j = \mathfrak{R}_{\phi_0}(\im f^j)$ and $\im f^j  = \mathfrak{L}_{\phi_0}(\ker f^j)$.  For other dualities, use Theorem~\ref{thm:char-subgroups} and Proposition~\ref{prop:FiltrationCharacteristic}. 
\end{proof}

\begin{rem}  \label{rem:KerImExample}
When $A$ is an elementary abelian $p$-group, the filtrations \eqref{eqn:TwoFiltrations} collapse, with $N=1$:  $A = \ker f$ and $0 = \im f$.  In contrast, when $A = \Z/p^\ell \Z$,  $\im f^j = p^j \Z/p^\ell \Z$, and $\ker f^j = p^{\ell - j} \Z/p^\ell \Z$.

Suppose $A = \Z/2\Z \oplus \Z/4 \Z$.  The subgroups $\ell_{\infty}$ and $S$ of Example~\ref{ex:4+2} are exactly $\ell_\infty = \im f$ and $S = \ker f$.
\end{rem}

\section{Congruence}  \label{sec:congruence}
There is an equivalence relation on dualities that generalizes the congruence of matrices and symmetric bilinear forms over finite prime fields.

\begin{defn}
Two dualities $\phi_1, \phi_2: A \ra \ah$ are \emph{congruent}, written $\phi_1 \simeq \phi_2$, if there exists an automorphism $\tau \in \aut(A)$ such that $\phi_2$ equals the composition
\[  \xymatrix{A \ar[r]^\tau & A \ar[r]^{\phi_1} & \ah \ar[r]^{\tau^*} & \ah }.  \]
\end{defn}
The condition for being congruent means, for all $a, a' \in A$, that
\begin{equation}  \label{eqn:congruence}
\Phi_2(a,a') = \pair{\phi_2(a)}{a'} = \pair{\phi_1(a \tau)}{a' \tau} = \Phi_1(a \tau, a' \tau).
\end{equation}
When $\phi_1 \simeq \phi_2$, $\phi_1$ is symmetric if and only if $\phi_2$ is symmetric.

\begin{ex}
For a prime $p$, let $A$ be an elementary abelian $p$-group of rank $n$, say $A = \F_p^n$.  In Example~\ref{ex:elem-ab-p-groups}, a duality $\phi_0$ of $A$ is defined by $\pair{\phi_0(a)}{b} = \zeta_p^{a b^\top} \in \C^\times$, for $a, b \in A$ (thought of as row vectors).  Any other duality has the form $\phi = \phi_0 \circ \tau$ for some automorphism $\tau \in \aut(A) = \GL(n, \F_p)$.  Regarding $\tau$ as an invertible matrix, we then have $\pair{\phi(a)}{b} = \pair{\phi_0(a \tau)}{b} = \zeta_p^{a \tau b^\top}$.

If $\phi' = \phi_0 \circ \tau'$, $\tau' \in \aut(A)$, is another duality, then $\phi'$ is congruent to $\phi$ if there exists an automorphism $\sigma \in \aut(A)$ such that $\phi' = \sigma^* \circ \phi \circ \sigma$.  This means, for any $a, b \in A$, that
\begin{align*}
\zeta_p^{a \tau' b^\top} &= \pair{\phi'(a)}{b} = \pair{\sigma^*(\phi(a \sigma))}{b} \\
&= \pair{\phi(a \sigma)}{b \sigma} = \zeta_p^{a \sigma \tau (b \sigma)^\top} 
= \zeta_p^{a \sigma \tau \sigma^\top b^\top} .
\end{align*}
These equations hold for all $a, b \in A$ if and only if $\tau' = \sigma \tau \sigma^\top$.  That is, $\tau$ and $\tau'$ are congruent matrices.  Hidden in plain view in the equations above is
\[  \pair{\phi'(a)}{b} = \pair{\phi(a \sigma)}{b \sigma} , \quad a,b \in A . \]

Because the homomorphism $\F_p \ra \C^\times$ sending $r \in \F_p$ to $\zeta_p^r \in \C^\times$ is injective, inner products on $A=\F_p^n$ are the same as nondegenerate bilinear forms on $A$ with values in $\F_p$.

When $p=2$, there are well-known results that classify nondegenerate symmetric bilinear forms.  The form $I$ is represented by the $1 \times 1$ matrix $[1]$, and the form $H$ is represented by 
\[  H = \begin{bmatrix}
0 & 1 \\ 1 & 0
\end{bmatrix} . \]
Every nondegenerate symmetric bilinear form over $\F_2$ is congruent to a direct sum of copies of $I$ and $H$, with the relation that $I + H \simeq 3 I$.

For odd primes $p$, nondegenerate symmetric bilinear forms are of two types, both diagonal: $1, 1, \ldots, 1,1$ and $1, 1, \ldots, 1, \la$,
where $\la$, in the words of Robert Wilson, is `your favourite nonsquare' in $\F_p$.  (When $p \equiv 1 \bmod 4$, $-1$ is a square in $\F_p$.)
\end{ex}

When two dualities are congruent, the comparative structure of subgroups and their dual codes align.
\begin{thm}
Let $A$ be a finite abelian group.  Suppose $\phi_1 \simeq \phi_2$ are congruent dualities of $A$, with $\phi_2 = \tau^* \circ \phi_1 \circ \tau$ for some $\tau \in \aut(A)$.  For subgroups $H, K \subseteq A$, $K = \mathfrak{L}_{\phi_2}(H)$ if and only if $K \tau = \mathfrak{L}_{\phi_1}(H \tau)$.  Likewise, $K = \mathfrak{R}_{\phi_2}(H)$ if and only if $K \tau = \mathfrak{R}_{\phi_1}(H \tau)$.  
\end{thm}

\begin{proof}
All the claims follow from \eqref{eqn:congruence} and the size condition.
\end{proof}

\begin{cor}
Let $A$ be a finite abelian group.  Suppose $\phi_1 \simeq \phi_2$ are congruent dualities of $A$, with $\phi_2 = \tau^* \circ \phi_1 \circ \tau$ for some $\tau \in \aut(A)$.  For a subgroup $H \subseteq A$, $H$ is self-dual under $\phi_2$ if and only if $H \tau$ is self-dual under $\phi_1$.  The number of self-dual codes under $\phi_1$ equals the number of self-dual codes under $\phi_2$
\end{cor}

\begin{ex}  \label{ex:Klein4-bis}
Let $A = \F_2^2$.  There are six dualities of $A$ listed in Example~\ref{ex:Klein-4-group}.  Three of the symmetric dualities are congruent: $\phi_0 \simeq \phi_1 \simeq \phi_2$.  The symmetric duality $\phi_3$ is congruent only to itself.  The two nonsymmetric dualities are congruent: $\phi_4 \simeq \phi_5$.  The subgroups of $A$ of order $2$ and their dual codes are displayed in Example~\ref{ex:Klein4}.  The number of self-dual codes is the same for congruent dualities.
\end{ex}

\begin{ex}  \label{ex:congruence3}
Let $A = \F_3^2$;  then $\aut(A) = \GL(2, \F_3)$.  As $\size{\GL(2, \F_3)} = (3^2-1)(3^2-3) = 48$, there are $48$ dualities.  Calculations (I used SageMath) reveal that $18$ dualities are symmetric, and $30$ dualities are not symmetric.  Representatives of the congruence classes and the number of dualities in each congruence class are displayed below.
\[ \begin{array}{c|cc|cccc}
\text{representative} & \left[\begin{smallmatrix}
1 & 0 \\ 0 & 1
\end{smallmatrix}\right] & \left[\begin{smallmatrix}
1 & 0 \\ 0 & 2
\end{smallmatrix}\right] & \left[\begin{smallmatrix}
1 & 1 \\ 0 & 1
\end{smallmatrix}\right] & \left[\begin{smallmatrix}
0 & 2 \\ 1 & 0
\end{smallmatrix}\right] & \left[\begin{smallmatrix}
2 & 1 \\ 0 & 1
\end{smallmatrix}\right] & \left[\begin{smallmatrix}
2 & 2 \\ 0 & 2
\end{smallmatrix}\right] 
\\ \hline
\text{number} & 6 & 12 & 8 & 2 & 12 & 8 
\end{array} \]

The abelian group $A$ has four (necessarily cyclic) subgroups of order $3$.  Here they are, with a chosen generator: $\ell_0 = \langle 10 \rangle$, $\ell_1 = \langle 11 \rangle$, $\ell_2 = \langle 12 \rangle$, and $\ell_\infty = \langle 01 \rangle$.  For any duality, the dual codes of the $\ell_j$ will be some permutation of the $\ell_j$.  Here are the various dual codes for the representatives of the congurence classes given above.  Recall that the left and right dual codes will be the same when the duality is symmetric.
\[  \begin{array}{cc|cc|cc|cc|cc}
\phi & \tau&  \mathfrak{L}(\ell_0) & \mathfrak{R}(\ell_0) & \mathfrak{L}(\ell_1) & \mathfrak{R}(\ell_1) & \mathfrak{L}(\ell_2) & \mathfrak{R}(\ell_2) & \mathfrak{L}(\ell_\infty) & \mathfrak{R}(\ell_\infty) \\ \hline
\phi_0 & \left[\begin{smallmatrix}
1 & 0 \\ 0 & 1
\end{smallmatrix}\right] & \ell_\infty & \ell_\infty & \ell_2 & \ell_2 & \ell_1 & \ell_1 & \ell_0 & \ell_0 \\ 
\phi_1 & \left[\begin{smallmatrix}
1 & 0 \\ 0 & 2
\end{smallmatrix}\right] & \ell_\infty & \ell_\infty & \ell_1 & \ell_1 & \ell_2 & \ell_2 & \ell_0 & \ell_0 \\ \hline
\phi_2 & \left[\begin{smallmatrix}
1 & 1 \\ 0 & 1
\end{smallmatrix}\right] & \ell_\infty & \ell_2 & \ell_1 & \ell_1 & \ell_0 & \ell_\infty & \ell_2 & \ell_0 \\ 
\phi_3 & \left[\begin{smallmatrix}
0 & 2 \\ 1 & 0
\end{smallmatrix}\right] & \ell_0 & \ell_0 & \ell_1 & \ell_1 & \ell_2 & \ell_2 & \ell_\infty & \ell_\infty \\ 
\phi_4 & \left[\begin{smallmatrix}
2 & 1 \\ 0 & 1
\end{smallmatrix}\right] & \ell_\infty & \ell_1 & \ell_0 & \ell_2 & \ell_1 & \ell_\infty & \ell_2 & \ell_0 \\ 
\phi_5 & \left[\begin{smallmatrix}
2 & 2 \\ 0 & 2
\end{smallmatrix}\right] & \ell_\infty & \ell_2 & \ell_1 & \ell_1 & \ell_0 & \ell_\infty & \ell_2 & \ell_0 
\end{array}  \]

The calculations also reveal that $\phi^* \simeq \phi$ for all dualities $\phi$ of $A$.  For a duality $\phi$ of $A$, define $\bar{\phi}(a) = \phi(-a)$ for all $a \in A$.  Four of the congruence classes satisfy $\phi_i \simeq \bar{\phi}_i$, namely $i=0, 1, 3, 4$, while $\phi_2 \simeq \bar{\phi}_5$.  The latter explains why $\phi_2$ and $\phi_5$ give the same dual codes for every subgroup.
\end{ex}

We expand on the observation in Example~\ref{ex:congruence3} about dualities that give the same dual codes for every subgroup.  

Suppose $A$ is a finite abelian $p$-group.  Let $m$ be an integer that is relatively prime to $p$.  Given a duality $\phi$ of $A$, define $\phi^m$ by
\[  \pair{\phi^m(a)}{b} = \pair{\phi(a)}{b}^m , \quad a, b \in A . \]
That is, $\phi^m(a) = (\phi(a))^m$, for $a \in A$, where the right side is the multiplication in the group $\ah$.
One verifies that $\phi^m$ is a duality of $A$; in fact, $\phi^m = \phi \circ (m \id_A)$, where $m \id_A$ is the automorphism of $A$ sending $a \in A$ to $ma \in A$.

\begin{lem}  \label{lem:DualOfmid}
Let $A$ be a finite abelian $p$-group, and let $m$ be an integer that is relatively prime to $p$.  If $\phi$ is a duality of $A$, then $(\phi^m)^* = (\phi^*)^m$.  In particular, if $\phi_2 = \phi_1 \circ m \id_A$, then $\phi_2^* = \phi_1^* \circ m \id_A$.
\end{lem}

\begin{proof}
From Lemma~\ref{lem:SymDual}, for any $a, b \in A$,
\begin{align*}
\pair{(\phi^m)^*(a)}{b} &= \pair{\phi^m(b)}{a} = \pair{\phi(b)}{a}^m \\
&= \pair{\phi^*(a)}{b}^m = \pair{(\phi^*)^m(a)}{b} .  \qedhere
\end{align*}
\end{proof}

\begin{thm}
Let $A$ be a finite abelian $p$-group, and suppose $\phi_1, \phi_2$ are two dualities of $A$.  Then, 
\[  \mathfrak{L}_{\phi_2}(H) = \mathfrak{L}_{\phi_1}(H), \quad \mathfrak{R}_{\phi_2}(H) = \mathfrak{R}_{\phi_1}(H) . \]
hold for all subgroups $H \subseteq A$, if and only if $\phi_2 = \phi_1^m$ for some integer $m$ that is relatively prime to $p$.
\end{thm}

\begin{proof}
If $m$ is relatively prime to $p$, then $m \id_A$ is an automorphism of $A$ and leaves every subgroup of $A$ invariant.  By Lemma~\ref{lem:DualOfmid} and Proposition~\ref{prop:SameDuals-stab}, if $\phi_2 = \phi_1^m$, then $\phi_1$ and $\phi_2$ yield the same dual codes for every subgroup.

Conversely, suppose $\phi_2 = \phi_1 \circ \tau$ for some $\tau \in \aut(A)$, and suppose $\phi_1, \phi_2$ yield the same dual codes for every subgroup.  By Proposition~\ref{prop:SameDuals-stab}, $\tau$ must leave every subgroup of $A$ invariant.  We need to show that $\tau = m \id_A$ for some integer $m$ that is relatively prime to $p$.

By the fundamental theorem of finite abelian groups, there are integers $1 \leq e_1 \leq e_2 \leq \cdots \leq e_\ell$ such that $A$ is isomorphic to
\[  \Z/p^{e_1} \Z \oplus \Z/p^{e_2} \Z \oplus \cdots \oplus \Z/p^{e_\ell} \Z  . \]
Among the subgroups of $A$ are those of the form $0 \oplus \cdots \oplus H_i \oplus \cdots \oplus 0$, with $0$s in all but one position, and $H_i=\Z/p^{e_i} \Z$.  Because all such subgroups are left invariant by $\tau$, we conclude that $\tau = \tau_1 \oplus \tau_2 \oplus \cdots \oplus \tau_\ell$, where each $\tau_i$ is an automorphism of $\Z/p^{e_i} \Z$.  By the structure of $\Z/p^{e_i} \Z$, we know that each $\tau_i$ is multiplication by some integer $m_i$ that is relatively prime to $p$.  By considering cyclic subgroups generated by elements such as $(0, \ldots, 0, 1, 1, 0, \ldots, 0)$, with two adjacent nonzero entries, invariance implies that $m_{i+1} \equiv m_i \bmod p^{e_i}$.  Then set $m=m_\ell$.
\end{proof}

\section{MacWilliams identities}  \label{sec:MWIds}

There are several forms of the MacWilliams identities that are valid over finite abelian groups \cite{delsarte:linear-programming, MR3336966, wood:turkey}.  We examine two cases: the Hamming weight enumerator and the complete enumerator.  (We will defer discussing the symmetrized enumerator of a group action to another paper.)  Both cases make use of the Fourier transform and the Poisson summation formula.

Let $A$ be a finite abelian group, with character group $\ah$.  For an element $a = (a_1, a_2, \ldots, a_n) \in A^n$, define its \emph{Hamming weight} by $\wh(a) = \size{\{i: a_i \neq 0\}}$.  When the abelian group is written multiplicatively, as with $\ah$, the Hamming weight is $\wh(\pi) = \size{\{i: \pi_i \neq 1\}}$.  For an additive code $C \subseteq A^n$, its \emph{Hamming weight enumerator} is the following polynomial in $\C[X,Y]$:
\[  \hwe_C(X,Y) = \sum_{c \in C} X^{n - \wh(c)} Y^{\wh(c)} .  \]

To define the complete enumerator, let $\C[Z_a: a \in A]$ (written $\C[Z_*]$, for short) be a polynomial ring with $\size{A}$ indeterminates $Z_a$ indexed by $a \in A$.  For an additive code $C \subseteq A^n$, its \emph{complete enumerator} is the following polynomial in $\C[Z_*]$:
\[  \ce_C(Z_*) = \sum_{c \in C} \prod_{i=1}^n Z_{c_i} . \]

Continue to let $A$ be a finite abelian group, with character group $\ah$.  Let $V$ be a vector space over the complex numbers $\C$.  Define $F(A, V) = \{f: A \ra V\}$, the set of all functions from $A$ to $V$; $F(A,V)$ is also a vector space over $\C$ under point-wise addition and scalar multiplication of functions.  The \emph{Fourier transform} is a $\C$-linear transformation $F(A,V) \ra F(\ah,V)$ defined by 
\begin{equation}  \label{eqn:DefFT}
\ft{f}(\pi) = \sum_{a \in A} \pair{\pi}{a} f(a) , \quad f \in F(A,V), \quad \pi \in \ah .
\end{equation}

\begin{lem}  \label{lem:FTInversion}
The Fourier transform is invertible.  For $f \in F(A,V)$ and $a \in A$,
\[  f(a) = \frac{1}{\size{A}} \sum_{\pi \in \ah} \pair{\pi}{-a} \ft{f}(\pi) . \]
\end{lem}

\begin{proof}
Calculate, using \eqref{eqn:DefFT}, Corollary~\ref{cor:VanishSum}, and Proposition~\ref{prop:SizeCharGroup}:
\begin{align*}
\sum_{\pi \in \ah} \pair{\pi}{-a} \ft{f}(\pi) &= \sum_{\pi \in \ah} \pair{\pi}{-a} \sum_{b \in A} \pair{\pi}{b} f(b) \\
&= \sum_{b \in A} \left( \sum_{\pi \in \ah} \pair{\pi}{b-a} \right) f(b) = \size{A} f(a) .  \qedhere
\end{align*}
\end{proof}

\begin{thm}[Poisson summation formula]  \label{thm:PSF}
Suppose $H \subseteq A$ is a subgroup of a finite abelian group $A$.  If $f \in F(A,V)$, then 
\[  \sum_{a \in H} f(a) = \frac{1}{\size{(A:H)}} \sum_{\pi \in (\ah:H)} \ft{f}(\pi) . \]
\end{thm}

\begin{proof}
Sum the equation in Lemma~\ref{lem:FTInversion} over $a \in H$:
\begin{align*}
\size{A} \sum_{a \in H} f(a) &= \sum_{a \in H} \sum_{\pi \in \ah} \pair{\pi}{-a} \ft{f}(\pi) \\
&= \sum_{\pi \in \ah} \left(  \sum_{a \in H} \pair{\pi}{-a} \right)  \ft{f}(\pi) = \size{H} \sum_{\pi \in (\ah:H)} \ft{f}(\pi) ,
\end{align*}
using Proposition~\ref{prop:VanishingSums} and Corollary~\ref{cor:DoubleAnnihilator}.
\end{proof}

We will now apply the Poisson summation formula to prove the MacWilliams identities for the Hamming and complete enumerators.  In the Poisson summation formula, the abelian group will be $A^n$ and the subgroup will be the additive code $C \subseteq A^n$.  The function $f: A^n \ra V$ will be $f: A^n \ra \C[X,Y]$, $f(x) = X^{n-\wh(x)} Y^{\wh(x)}$, in the case of the Hamming weight enuemrator, and $f: A^n \ra \C[Z_*]$, $f(x) = \prod_{i=1}^n Z_{x_i}$, in the case of the complete enumerator.

Both functions $f: A^n \ra V$ have a special form.  The vector space $V$ is actually a commutative complex algebra $\mathscr{A}$, and the function is a product of functions from $A$ to $\mathscr{A}$.  To be specific, in the Hamming case, let $g: A \ra \C[X,Y]$ be $g(a) = X^{1-\wh(a)} Y^{\wh(a)}$.  In the complete case, let $g: A \ra \C[Z_*]$ be $g(a) = Z_a$.  Then, in each case, for $x = (x_1, x_2, \ldots, x_n) \in A^n$, $f(x) = \prod_{i=1}^n g(x_i)$.  Because of this special form, the Fourier transform is easy to calculate.

\begin{lem}  \label{lem:FTOfProduct}
Let $A$ be a finite abelian group and $\mathscr{A}$ be a commutative complex algebra.  Suppose there are functions $f_i: A \ra \mathscr{A}$, $i=1, 2, \ldots, n$, such that $f: A^n \ra \mathscr{A}$ satisfies $f(x) = \prod_{i=1}^n f_i(x_i)$, for $x = (x_1, x_2, \ldots, x_n) \in A^n$.  Then, for $\pi = (\pi_1, \pi_2, \ldots, \pi_n) \in \ah^n$, 
\[  \ft{f}(\pi) = \prod_{i=1}^n \ft{f}_i(\pi_i) . \]
\end{lem}

\begin{proof}
This is a calculation, using Lemma~\ref{lem:Products}:
\begin{align*}
\ft{f}(\pi) &= \sum_{x \in A^n} \pair{\pi}{x} f(x) = \sum_{x \in A^n} \prod_{i=1}^n \left( \pair{\pi_i}{x_i} f_i(x_i) \right) \\
&= \prod_{i=1}^n \left( \sum_{x_i \in A} \pair{\pi_i}{x_i} f_i(x_i) \right) =  \prod_{i=1}^n \ft{f}_i(\pi_i) . \qedhere
\end{align*}
\end{proof}

We now calculate the Fourier transforms of the one-variable functions in each case.

\begin{lem}  \label{lem:FTforHamming}
Let $A$ be a finite abelian group.  If $g: A \ra \C[X,Y]$ is given by $g(a) = X^{1-\wh(a)} Y^{\wh(a)}$, then, for $\pi \in \ah$,
\[  \ft{g}(\pi) = \begin{cases}
X + (\size{A}-1) Y, & \pi = 1, \\
X-Y, & \pi \neq 1 .
\end{cases}  \]
\end{lem}

\begin{proof}
Calculate, using $\pi(0)=1$ and Corollary~\ref{cor:VanishSum}:
\begin{align*}
\ft{g}(\pi) &= \sum_{a \in A} \pair{\pi}{a} X^{1-\wh(a)} Y^{\wh(a)} = X + \sum_{a \neq 0} \pair{\pi}{a} Y \\
&= \begin{cases}
X + (\size{A}-1) Y, & \pi = 1, \\
X-Y, & \pi \neq 1. 
\end{cases}  \qedhere
\end{align*}
\end{proof}

\begin{lem}
Let $A$ be a finite abelian group.  If $g: A \ra \C[Z_*]$ is given by $g(a) = Z_a$, then, for $\pi \in \ah$,
\[  \ft{g}(\pi) = \sum_{a \in A} \pair{\pi}{a} Z_a .  \]
\end{lem}

\begin{proof}
This is \eqref{eqn:DefFT}.
\end{proof}

We now assemble all the pieces.

\begin{thm}[MacWilliams identities]  \label{thm:AhatFormMWids}
Let $A$ be a finite abelian group.  If $C \subseteq A^n$ is an additive code, then 
\begin{align*}
\hwe_C(X,Y) &= \frac{1}{\size{(\ah^n:C)}} \hwe_{(\ah^n:C)}(X+(\size{A}-1)Y, X-Y) , \\
\ce_C(Z_*) &= \frac{1}{\size{(\ah^n:C)}} \left. \ce_{(\ah^n:C)}(\mathcal{Z}_*)\right|_{\mathcal{Z}_\pi \leftarrow \sum_{a \in A} \pair{\pi}{a} Z_a} .
\end{align*}
\end{thm}

\begin{proof}
As described earlier, one applies the Poisson summation formula, Theorem~\ref{thm:PSF}, to $C \subseteq A^n$, with $f(x) = X^{n-\wh(x)} Y^{\wh(x)}$ in the Hamming case and $f(x) = \prod_{i=1}^n Z_{x_i}$ in the complete case.  The key step is to recognize that the form of the factorization of $\ft{f}(\pi)$ given by Lemma~\ref{lem:FTOfProduct} depends exactly on $f(\pi)$.  Thus, the right side of the Poisson summation formula has the form of an enumerator.  The appropriate $\ft{g}(\pi_i)$ is then substituted.
\end{proof}

Note that the Hamming weight enumerator is obtained by substituting $X$ for $Z_0$ and $Y$ for each $Z_a$, $a \neq 0$, in the complete enumerator.  The resulting simplification of terms mimics the proof of Lemma~\ref{lem:FTforHamming}.

\begin{rem}  \label{rem:SwapSides}
Because of double duality, Lemma~\ref{cor:DoubleAnnihilator}, the roles of the additive code $C \subseteq A^n$ and its annihilator $(\ah^n:C) \subseteq \ah^n$ can be reversed.  Note the subtle difference in the substitutions in the complete case, which is a consequence of \eqref{eqn:DefnEval}.
\begin{align*}
\hwe_{(\ah^n:C)}(X,Y) &= \frac{1}{\size{C}} \hwe_C(X+(\size{A}-1)Y, X-Y) , \\
\ce_{(\ah^n:C)}(\mathcal{Z}_*) &= \frac{1}{\size{C}} \left. \ce_C(Z_*)\right|_{Z_a \leftarrow \sum_{\pi \in \ah} \pair{\pi}{a} \mathcal{Z}_\pi} .
\end{align*}
\end{rem}

Armed with Theorem~\ref{thm:AhatFormMWids}, we now fix a duality $\phi$ of $A$, extend it to $A^n$, and pull back the results to $A^n$.

\begin{thm}[MacWilliams identities]  \label{thm:PulledBackMWids}
Let $A$ be a finite abelian group, and let $\phi$ be a duality of $A$, extended to $A^n$.  If $C \subseteq A^n$ is an additive code, then 
\begin{align*}
\hwe_C(X,Y) &= \frac{1}{\size{\mathfrak{L}_\phi(C)}} \hwe_{\mathfrak{L}_\phi(C)}(X+(\size{A}-1)Y, X-Y) , \\
\hwe_C(X,Y) &= \frac{1}{\size{\mathfrak{R}_\phi(C)}} \hwe_{\mathfrak{R}_\phi(C)}(X+(\size{A}-1)Y, X-Y) , \\
\ce_C(Z_*) &= \frac{1}{\size{\mathfrak{L}_\phi(C)}} \left. \ce_{\mathfrak{L}_\phi(C)}(\mathfrak{Z}_*)\right|_{\mathfrak{Z}_b \leftarrow \sum_{a \in A} \Phi(b,a) Z_a} ,  \\
\ce_C(Z_*) &= \frac{1}{\size{\mathfrak{R}_\phi(C)}} \left. \ce_{\mathfrak{R}_\phi(C)}(\mathfrak{Z}_*)\right|_{\mathfrak{Z}_b \leftarrow \sum_{a \in A} \Phi(a,b) Z_a} .
\end{align*}
\end{thm}

\begin{proof}
Suppose $C \subseteq A^n$ is an additive code.  The duality $\phi$, extended coordinatewise to $A^n$, is an isomorphism from $A^n$ to $\ah^n$ that takes $\mathfrak{L}_\phi(C)$ to $(\ah^n:C)$, by Lemma~\ref{lem:LeftRightDuals}.  Because $\phi$ is extended coordinatewise, $\phi$ preserves the Hamming weight, so that $\wh(x) = \wh(\phi(x))$ for all $x \in A^n$.  This implies the first equation in the theorem.  The second equation follows from the same argument applied to $\phi^*$.

For the complete enumerators, we use the indeterminates $Z_*$ for $C$, $\mathcal{Z}_*$ for $(\ah^n:C)$, and $\mathfrak{Z}_*$ for the dual codes.  Under $\phi: A^n \ra \ah^n$, $\mathfrak{Z}_a$ will correspond to $\mathcal{Z}_{\phi(a)}$, while under $\phi^*: A^n \ra \ah^n$, $\mathfrak{Z}_a$ will correspond to $\mathcal{Z}_{\phi^*(a)}$.  Thus, using $\phi$, $\mathfrak{L}_\phi(C)$ corresponds to $(\ah^n:C)$, and the substitution is $\mathfrak{Z}_b \leftarrow \sum_{a \in A} \pair{\phi(b)}{a} Z_a = \sum_{a \in A} \Phi(b,a) Z_a$.  In contrast, using $\phi^*$, $\mathfrak{R}_\phi(C)$ corresponds to $(\ah^n:C)$, and the substitution is $\mathfrak{Z}_b \leftarrow \sum_{a \in A} \pair{\phi^*(b)}{a} Z_a = \sum_{a \in A} \Phi^*(b,a) Z_a = \sum_{a \in A} \Phi(a,b) Z_a$.
\end{proof}

Now we reverse the roles of the additive code $C \subseteq A^n$ and its dual codes $\mathfrak{L}_\phi(C), \mathfrak{R}_\phi(C) \subseteq A^n$, using Proposition~\ref{prop:DualProperties}.

\begin{cor}
Let $A$ be a finite abelian group, and let $\phi$ be a duality of $A$, extended to $A^n$.  If $C \subseteq A^n$ is an additive code, then 
\begin{align*}
\hwe_{\mathfrak{L}_\phi(C)}(X,Y) &= \hwe_{\mathfrak{R}_\phi(C)} (X,Y) \\
&= \frac{1}{\size{C}} \hwe_C (X+(\size{A}-1)Y, X-Y) , \\
\ce_{\mathfrak{L}_\phi(C)}(\mathfrak{Z}_*) &= \frac{1}{\size{C}} \left. \ce_C(Z_*)\right|_{Z_b \leftarrow \sum_{a \in A} \Phi(a,b) \mathfrak{Z}_a} ,  \\
\ce_{\mathfrak{R}_\phi(C)}(\mathfrak{Z}_*) &= \frac{1}{\size{C}} \left. \ce_C (Z_*)\right|_{Z_b \leftarrow \sum_{a \in A} \Phi(b,a) \mathfrak{Z}_a} .
\end{align*}
\end{cor}

\begin{proof}
As in the proof of Theorem~\ref{thm:PulledBackMWids}, both $\phi$ and $\phi^*$ preserve the Hamming weight between the dual codes $\mathfrak{L}_\phi(C), \mathfrak{R}_\phi(C)$ and the annihilator $(\ah^n:C)$, so their Hamming weight enumerators are equal.  The Hamming result then follows from Remark~\ref{rem:SwapSides}.

For the complete enumerator, $\phi$ matches $\mathfrak{Z}_a$ with $\mathcal{Z}_{\phi(a)}$, while $\phi^*$ matches $\mathfrak{Z}_a$ with $\mathcal{Z}_{\phi^*(a)}$.  The substitution $Z_b \leftarrow \sum_{\pi \in \ah} \pair{\pi}{b} \mathcal{Z}_\pi$ from Remark~\ref{rem:SwapSides} then becomes $Z_b \leftarrow \sum_{a \in A} \pair{\phi(a)}{b} \mathfrak{Z}_a = \sum_{a \in A} \Phi(a,b) \mathfrak{Z}_a$ for $\phi$ (contrary to \cite[Theorem~3.3]{MR3789447}) and $Z_b \leftarrow \sum_{a \in A} \pair{\phi^*(a)}{b} \mathfrak{Z}_a = \sum_{a \in A} \Phi(b,a)  \mathfrak{Z}_a$ for $\phi^*$.
\end{proof}

\begin{rem}
Versions of the MacWilliams identities for both complete and Hamming joint enumerators appear in \cite{STD:JointEnums}.
\end{rem}

\section{A comment about finite rings}  \label{sec:FinRings}
Suppose a finite abelian group $A$ has the additional algebraic structure of being the additive group of a finite ring $R$ (with or without a multiplicative identity).  Then the character group $\hr$ has the structure of an $R$-bimodule.  For $r,s \in R$ and $\pi \in \hr$, the scalar multiplications are written in multiplicative form:
\[  \pair{{}^r \pi}{s} = \pair{\pi}{sr}, \quad \pair{\pi^r}{s} = \pair{\pi}{rs} . \]

Define a duality $\phi$ of $R$ to be \emph{linear} when $\phi: R \ra \hr$ is a homomorphism of left (or right) $R$-modules.  Linear dualities do not always exist.  If the finite ring has a $1$, then a linear duality exists if and only if $R$ is a Frobenius ring \cite[Theorem~3.10]{wood:duality}.  If $R$ is Frobenius and $\phi$ is a left linear duality of $R$, then the image $\chi = \phi(1) \in \hr$ is a \emph{generating character} in the sense that $\phi(r) = {}^r \chi$ is an isomorphism $R \ra \hr$ of left $R$-modules.  One then shows that $\phi^*(r) = \chi^r$ is an isomorphism $R \ra \hr$ of right $R$-modules.  There is a well-developed theory of linear dualities over Frobenius rings; see \cite[\S 12]{wood:turkey}.

Less is known when $R$ is a finite rng (i.e., a finite ring not necessarily having a multiplicative identity: no `i').  However, one situation is understood: when $R$ has a generating character.

\begin{thm}  \label{thm:GenCharRng}
Let $R$ be a finite ring, not necessarily having a multiplicative identity.  Suppose there exists a character $\chi \in \hr$ such that $R \ra \hr$, $r \mapsto \chi^r$, resp., $r \mapsto {}^r \chi$, is an isomorphism of right, resp., left, $R$-modules.  Then $R$ has a multiplicative identity, and $R$ is Frobenius. 
\end{thm}

\begin{proof}
By surjectivity, there exists an element $e \in R$ such that $\chi^e = \chi$.  We claim that $e$ is a multiplicative identity.  For any $r \in R$, right multiply by $r$: $\chi^{er} = \chi^r$.  Injectivity implies $er=r$.  Now consider $\chi^{re} = (\chi^r)^e$.  For any $s \in R$, $\pair{\chi^{re}}{s} = \pair{(\chi^r)^e}{s} = \pair{\chi^r}{es} = \pair{\chi^r}{s}$.  Thus, $\chi^{re}=\chi^r$, and injectivity implies $re=r$.  Then $R$ is Frobenius by \cite[Theorem~3.10]{wood:duality}.  The left module version is similar.
\end{proof}

There exist linear dualities of finite rngs that are not of the form in Theorem~\ref{thm:GenCharRng}.

\begin{ex}
Consider the rng $I$ defined as the ideal $(x) \subset \F_2[x]/(x^3)$.  This is the $p=2$ version of the rng $I$ listed in \cite{MR1240670}.  Here are the addition and multiplication tables.
\[ \begin{array}{c|cccc}
+ & 0 & x & x+x^2 & x^2 \\ \hline
0 & 0 & x & x+x^2 & x^2 \\
x & x & 0 & x^2 & x+x^2 \\
x+x^2 & x+x^2 & x^2 & 0 & x \\
x^2 & x^2 & x+x^2 & x & 0
\end{array}
\]
\[ \begin{array}{c|cccc}
\times & 0 & x & x+x^2 & x^2 \\ \hline
0 & 0 & 0 & 0 & 0 \\
x & 0 & x^2 & x^2 & 0 \\
x+x^2 & 0 & x^2 & x^2 & 0 \\
x^2 & 0 & 0 & 0 & 0
\end{array}
\]
The additive group is a Klein $4$-group; multiplication is commutative.

The character module $\charac{I}$ has the following elements and scalar multiplications.
\[  \begin{array}{c|rrrr|rrrr}
s & \pi_0(s) & \pi_1(s) & \pi_2(s) & \pi_3(s) & \pi_0^s & \pi_1^s & \pi_2^s & \pi_3^s \\ \hline
0 & 1 & 1 & 1 & 1& \pi_0 & \pi_0 & \pi_0 & \pi_0 \\
x & 1 & 1 & -1 & -1 & \pi_0 & \pi_3 & \pi_3 & \pi_0 \\
x+x^2 & 1 & -1 & 1 & -1 & \pi_0 & \pi_3 & \pi_3 & \pi_0 \\
x^2 &  1 & -1 & -1 & 1 & \pi_0 & \pi_0 & \pi_0 & \pi_0 
\end{array}  \]
Then $f: I \ra \charac{I}$ defined by 
\[  \begin{array}{c|cccc}
s & 0 & x & x+x^2 & x^2 \\ \hline
f(s) & \pi_0 & \pi_1 & \pi_2 & \pi_3 
\end{array} \]
is seen to be an isomorphism of $I$-modules.  (One could also interchange the roles of $x$ and $x+x^2$ (or of $\pi_1$ and $\pi_2$).)
\end{ex}

In a future paper, I plan to discuss the MacWilliams identities for the symmetrized enumerator associated to a group action on $A$, as well as to discuss dualities in the context of module alphabets over finite rings and how dualities interact with notions of equivalence of codes.

\def\cprime{$'$} \def\cprime{$'$} \def\cprime{$'$} \def\cprime{$'$}
  \def\cprime{$'$}
\providecommand{\bysame}{\leavevmode\hbox to3em{\hrulefill}\thinspace}
\providecommand{\MR}{\relax\ifhmode\unskip\space\fi MR }
\providecommand{\MRhref}[2]{%
  \href{http://www.ams.org/mathscinet-getitem?mr=#1}{#2}
}
\providecommand{\href}[2]{#2}

\end{document}